%% file: main.tex
\documentclass[11pt,letterpaper]{article}

\usepackage[margin=1in]{geometry}
\usepackage[T1]{fontenc}
\usepackage{lmodern}
\usepackage{microtype}
\usepackage{amsmath,amssymb,amsthm,mathtools}
\usepackage{aliascnt}
\usepackage{booktabs}
\usepackage{graphicx}
\usepackage{float}
\usepackage{enumitem}
\usepackage{placeins}
\usepackage{needspace}
\usepackage{xcolor}
\usepackage{url}
\usepackage[hidelinks]{hyperref}
\usepackage[nameinlink,capitalise,noabbrev]{cleveref}

\input{macros}

\title{Synchronization Strings over the Optimal Alphabet}
\author{Huibo Xu$^{1,\dagger}$ \quad Shi Fu$^{1,\dagger}$ \quad
Youming Qiao$^{2}$ \quad Dacheng Tao$^{1}$\\[0.5em]
\small $^{1}$Nanyang Technological University, Singapore\\
\small $^{2}$University of Technology Sydney, Sydney, Australia\\[0.35em]
\small $^{\dagger}$These authors contributed equally to this work.}
\date{}

\begin{document}
\hypersetup{pageanchor=false}
\maketitle
\thispagestyle{empty}

\begin{abstract}
  \input{sections/00_abstract}
\end{abstract}

\clearpage
\hypersetup{pageanchor=true}
\pagenumbering{arabic}

\input{sections/01_introduction}
\input{sections/02_preliminaries}
\input{sections/03_brinkhuis_source}
\input{sections/04_ternary_proof}
\input{sections/05_structural_theorem}
\input{sections/05_quantitative_summary}
\input{sections/06_abundance_algorithm}
\input{sections/07_discussion}

\clearpage
\appendix
\linespread{1.03}\selectfont
\setlength{\parskip}{1.5pt plus 0.4pt minus 0.2pt}
\setlength{\jot}{4pt}
\input{appendix/C_general_theorem}
\input{appendix/E_extremal_square_free}
\input{appendix/D_all_length_circles}
\input{sections/05_quantitative_strengthening}
\input{appendix/B_numeric_details}
\input{appendix/G_long_distance}

\clearpage
\linespread{1}\selectfont
\setlength{\parskip}{0pt plus 1pt}
\setlength{\jot}{3pt}
\bibliographystyle{alpha}
\bibliography{references}

\end{document}

%% file: macros.tex
\newcommand{\eps}{\varepsilon}

\newcommand{\alphabet}{\Sigma}
\newcommand{\ternary}{\{0,1,2\}}
\newcommand{\bits}{\{0,1\}}

\newcommand{\cE}{\mathcal{E}}
\newcommand{\cT}{\mathcal{T}}
\newcommand{\cC}{\mathcal{C}}

\newcommand{\vbl}{\operatorname{vbl}}
\DeclareMathOperator{\ED}{ED}
\DeclareMathOperator{\LCS}{LCS}

\newtheorem{theorem}{Theorem}[section]
\newaliascnt{lemma}{theorem}
\newtheorem{lemma}[lemma]{Lemma}
\aliascntresetthe{lemma}
\newaliascnt{proposition}{theorem}
\newtheorem{proposition}[proposition]{Proposition}
\aliascntresetthe{proposition}
\newaliascnt{corollary}{theorem}
\newtheorem{corollary}[corollary]{Corollary}
\aliascntresetthe{corollary}
\newaliascnt{claim}{theorem}

\aliascntresetthe{claim}
\theoremstyle{definition}
\newaliascnt{definition}{theorem}
\newtheorem{definition}[definition]{Definition}
\aliascntresetthe{definition}
\newaliascnt{remark}{theorem}

\aliascntresetthe{remark}

\crefname{theorem}{Theorem}{Theorems}
\crefname{lemma}{Lemma}{Lemmas}
\crefname{proposition}{Proposition}{Propositions}
\crefname{corollary}{Corollary}{Corollaries}
\crefname{definition}{Definition}{Definitions}
\crefname{claim}{Claim}{Claims}
\crefname{remark}{Remark}{Remarks}

\setlist[itemize]{leftmargin=1.5em,itemsep=0.25em,topsep=0.4em}
\setlist[enumerate]{leftmargin=1.8em,itemsep=0.25em,topsep=0.4em}

\newcommand{\mainboundarygap}{\frac{1}{216}}
\newcommand{\mainthreshold}{\frac{215}{216}}
\newcommand{\maineps}{\frac{216}{217}}
\newcommand{\baseeta}{\frac{1}{1000}}
\newcommand{\basecutoff}{1000}
\newcommand{\basegap}{\frac{1}{2002}}
\newcommand{\baseeps}{\frac{2001}{2002}}
\providecommand{\Description}[1]{}

%% file: sections/00_abstract.tex
Synchronization strings provide deterministic position labels for recovering
coordinates after insertions and deletions.  Haeupler and Shahrasbi introduced
these objects, and subsequent work proved that four symbols suffice for some
fixed parameter $\eps<1$, whereas two symbols cannot support arbitrarily long
synchronization strings.  We resolve the remaining ternary case: every length
admits a ternary $2001/2002$-synchronization string.  Thus three is the exact
minimum constant alphabet size.  A computer-assisted refinement based on a
larger 54-uniform family yields ternary $\eps$-synchronization strings for
every $\eps>215/216$.

The previous four-symbol construction uses a ternary square-free backbone
to exclude short repetitions and a fourth symbol to carry long-range
synchronization marks.  Our main technical contribution is a local-entropy
transfer theorem: every square-free block-local source with a positive
interval conditional min-entropy rate supports synchronization strings with a fixed gap.
We instantiate this theorem using occurrence-wise branching in a Brinkhuis
family.  Every outcome remains ternary and square-free, while every long
interval retains linear conditional min-entropy after all choices outside
it are exposed.  A deletion-ball estimate converts this entropy into
an exponentially small probability of a near-complete common subsequence between
adjacent intervals, and an asymmetric
Lov\'asz Local Lemma enforces all interval constraints simultaneously.
The same framework also yields exponentially many valid words,
synchronization circles, and synchronization within a class of extremal
square-free words.  Adding constraints on distant intervals gives a Las Vegas
construction in expected $O(n^2\log^3(n+2))$ time.

%% file: sections/01_introduction.tex
\section{Introduction}

Synchronization strings restore coordinates after insertions and deletions
by attaching a deterministic label to each position.  Four labels were known
to suffice for a fixed synchronization gap, and two were known not to
suffice; whether three labels suffice was left open.  We prove that they do.
The problem joins two themes: coding against insertions and deletions, where
one needs a low-cost index that survives a loss of alignment, and pattern
avoidance, where three symbols are exactly what is needed to exclude repeated
blocks.

\subsection{Motivation and previous work in insertion--deletion coding}

An insertion or deletion does more than corrupt one symbol: it shifts the
coordinate system of every symbol that follows.  If one record disappears
from an indexed stream, then all later records appear one position too early.
Ordinary error-correcting redundancy can protect values once their positions
are known, but it does not by itself restore these lost coordinates.

Synchronization strings were introduced by Haeupler and Shahrasbi as a
deterministic coordinate system for this setting
\cite{HaeuplerShahrasbi2017}.  A fixed label is attached to each transmitted
position; even after insertions and deletions, the surviving labels allow the
receiver to recover approximate positions and reduce synchronization errors
to erasures and corruptions.  This reusable indexing layer has since been
used in coding, document exchange, channel simulation, interactive
communication, and edit-distance indexing
\cite{HaeuplerShahrasbi2018,Haeupler2019DocumentExchange,
HaeuplerShahrasbiVitercik2018,HaeuplerRubinsteinShahrasbi2019,
HaeuplerShahrasbiSurvey2021}.

The combinatorial definition is local.  For words $x$ and $y$, let
$\ED(x,y)$ be their insertion--deletion distance.

\begin{definition}[Synchronization string]\label{def:sync}
For $0\le\eps<1$, a word $S\in\alphabet^n$ is an
$\eps$-synchronization string if, for every $1\le i<j<k\le n+1$,
\begin{equation}\label{eq:intro-definition}
  \ED(S[i,j),S[j,k))>(1-\eps)(k-i).
\end{equation}
The number $1-\eps$ is the \emph{synchronization gap}.
\end{definition}

The alphabet must be fixed independently of the length of $S$.
The interval inequality rules out the basic ambiguity created by a
position shift.  Suppose a received subsequence could plausibly be aligned
both to positions just before a boundary $j$ and to positions just after it.
Those two explanations would exhibit a long common subsequence between the
adjacent label intervals $S[i,j)$ and $S[j,k)$.  The synchronization
inequality says that every such cross-boundary match must discard a fixed
fraction of the two intervals.  Consequently, a decoder can charge a
substantial amount of alignment loss to each persistent displacement and
convert insertions and deletions into a controlled number of ordinary symbol
errors or erasures.  Our focus is the combinatorial object that makes these
reductions possible; the cited synchronization-string literature develops
the corresponding decoders and coding applications.

The inequality must hold simultaneously at every boundary and at every
scale, from single symbols to intervals spanning a constant fraction of the
word.  A construction must therefore combine short-range pattern avoidance
with long-range separation on the same alphabet.

\subsection{Pattern avoidance and the extremal alphabet question}

A \emph{square} is a nonempty word repeated twice consecutively, $XX$; for
example, $012012$ is the square of $012$.  A word is \emph{square-free} if it
contains no such factor.  Synchronization
strings with $\eps<1$ must be square-free, since the two copies of a square
have edit distance zero.  Thue's classical theorem shows that arbitrarily
long square-free words exist over three symbols, already suggesting why three
is the natural lower limit.

Haeupler and Shahrasbi established the first constant-alphabet synchronization
strings.  Cheng, Haeupler, Li, Shahrasbi, and Wu subsequently studied the
extremal alphabet question: they proved that four symbols suffice for some
fixed $\eps<1$, whereas binary alphabets cannot support arbitrarily long
synchronization strings
\cite[Section~1.3.2 and Theorem~6.4]{ChengEtAl2019}.  Thus the minimum alphabet was
known to be either three or four.  The subsequent survey still listed the
ternary case as open~\cite[Section~3.1]{HaeuplerShahrasbiSurvey2021}; the
later four-symbol result discussed below did not change this.

The distinction matters because the index alphabet is side information
attached to every transmitted symbol.  In a packed representation, moving
from four labels to three reduces the idealized label cost from $2$ to
$\log_2 3$ bits per position.  The total coding rate also depends on the
synchronization parameter and the outer code; here we determine the alphabet
threshold.  Combinatorially, the question is whether the first alphabet that
escapes the local repetition obstruction already supports the stronger
linear separation required at every scale.

Square-freeness alone is nevertheless too weak.  Let $X$ and $Y$ be adjacent
words of the same length $\ell$, and let $\LCS(X,Y)$ denote the length of
their longest common subsequence.  Square-freeness says only that $X\ne Y$,
and hence $\LCS(X,Y)\le \ell-1$.  This loss of one symbol can occur even in
the square-free word $01020120\,21020120$, whose two displayed length-eight
halves have a common
subsequence of length seven.  Thus excluding $XX$ does not exclude adjacent
blocks that become identical after one deletion from each side.
The relative deficit $1/\ell$ vanishes with the interval length.  By
contrast, synchronization requires a fixed linear deficit
$\LCS(X,Y)\le(1-\eta)\ell$.

The extremal question asks whether this stronger separation is possible over
three symbols for one fixed $\eps<1$.  A negative answer would separate the
alphabet threshold for avoiding exact repetitions from that for preventing
approximate repetitions.  Our main theorem shows that the thresholds coincide:
the first alphabet large enough to avoid squares also supports a linear
edit-distance gap at every scale.

\subsection{Our results}

\begin{theorem}[Ternary synchronization strings and optimality]
\label{thm:ternary-intro}
For every integer $n\ge1$, there is a ternary
$\baseeps$-synchronization string of length $n$.  Consequently, three is the
minimum constant alphabet size supporting synchronization strings of all
lengths for some fixed $\eps<1$.
\end{theorem}

The proof uses the corrected 18-uniform Brinkhuis family, comprising two
images for each of the three outer letters.  Square-freeness controls bounded
scales, while the occurrence-wise choices separate long intervals.  The
binary impossibility theorem of Cheng et
al.~\cite[Section~1.3.2]{ChengEtAl2019} then makes
the alphabet size optimal.

\subsubsection*{From four symbols to three}
The construction of Cheng, Haeupler, Li, Shahrasbi, and Wu separates the two
scales by using two distinct resources~\cite[Section~6.2]{ChengEtAl2019}: a ternary
square-free backbone rules out exact repetitions and handles the bounded
scales.  Occurrences of one backbone letter are then split into two output
symbols according to a binary weak synchronization string; this adds a fourth
symbol and stores the long-range marking independently of the backbone.  Our
construction remains ternary throughout.  Instead, it chooses independently
among several square-free-preserving images at each occurrence of an outer
letter.  After conditioning outside one interval, the remaining choices give
anti-concentration against a high-LCS match on the adjacent interval.  Thus
the issue resolved here is the one left open by the four-symbol method:
maintaining ternary square-freeness while simultaneously obtaining a linear
edit-distance gap at every scale.  Cheng et al. also show that
iterating a fixed uniform morphism cannot by itself generate arbitrarily long
synchronization strings, so the occurrence-wise variation is essential to
how our construction escapes that obstruction.

The difficulty is to obtain both properties in the same random source.
An independently random ternary word has long-range
anti-concentration, but with overwhelming probability it contains short
squares.  Conditioning on global square-freeness destroys the product
structure used by the usual interval-by-interval local-lemma analysis.

A fixed square-free word, or an iterate of a fixed square-free morphism,
guarantees local pattern avoidance but supplies no random choices with which
to rule out approximate copies having LCS $\ell-o(\ell)$.  Fixed uniform
morphisms also face the obstruction proved in the four-symbol work, while
encoding an external binary mark by changing ternary letters need not
preserve square-freeness.

Brinkhuis families combine the two properties.  After the outer word is
fixed they expose independent local choices, every combination of choices
preserves ternary square-freeness, and a long interval contains linearly many
unrevealed choices.  The corrected 18-uniform family originating with Ekhad
and Zeilberger permits two independent images at every outer-letter
occurrence while preserving square-freeness for every choice vector
\cite{EkhadZeilberger1998,EkhadZeilbergerErratum2001,Grimm2001,
SollamiDouglasLiebmann2016}.  These choices were introduced for counting
square-free words.  We condition on the choices outside an interval and use
the choices that remain inside as local randomness.  Comparing their
conditional min-entropy with an LCS deletion ball turns the branching
substitution into a synchronization mechanism.

\begin{figure}[H]
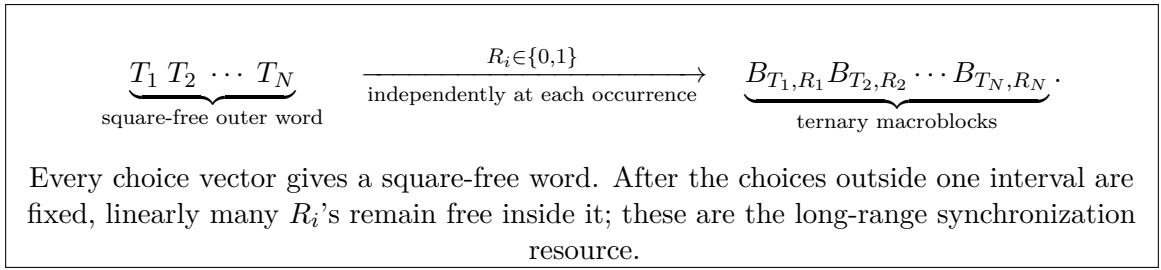

  \centering
  \fbox{\begin{minipage}{0.91\linewidth}
  \centering
  \[
    \underbrace{T_1\;T_2\;\cdots\;T_N}_{\text{square-free outer word}}
    \quad\xrightarrow[\text{independently at each occurrence}]
      {\ R_i\in\{0,1\}\ }
    \quad
    \underbrace{B_{T_1,R_1}B_{T_2,R_2}\cdots B_{T_N,R_N}}
      _{\text{ternary macroblocks}} .
  \]
  Every choice vector gives a square-free word.  After the choices outside
  one interval are fixed, linearly many $R_i$'s remain free inside it; these
  are the long-range synchronization resource.
  \end{minipage}}
  \caption{The resource used by the construction.  The output alphabet is
  always ternary: long-range separation is encoded by branching within the
  square-free language rather than by a fourth marking symbol.}
  \label{fig:intro-construction}
\end{figure}

\subsubsection*{Local entropy and further consequences}
Exponential abundance alone would not give the randomness needed here.
A family can contain exponentially many words while all of its variation
is concentrated in one region; even variation throughout the word may be
determined once the outside of an interval is revealed.  Neither possibility
is ruled out by a lower bound on the total number of square-free words.
To control a match across any prescribed boundary, we need randomness that
remains in the first interval after the second has been fixed.  The
occurrence-wise structure of a Brinkhuis family supplies precisely this
property: fixing all choices except those of complete blocks inside the
first interval leaves a fresh independent choice in each such block.

The useful measure of this randomness is conditional min-entropy.  It bounds
the probability of every individual interval value, so multiplying its
point-mass bound by the number of high-LCS competitors bounds the probability
of a bad match.  This motivates a \emph{block-local source}: independent
variables determine blocks of bounded
length, and every long interval retains linear conditional min-entropy when
all variables except those of its complete blocks are fixed.  Bounded block
lengths also control the event dependencies needed by the local lemma.

\begin{theorem}[Local-entropy transfer, informal]
\label{thm:structural-intro}
Every square-free block-local source with fixed block size and positive
interval conditional min-entropy rate contains a synchronization string for
one fixed parameter $\eps<1$.  Quantitatively, the achievable LCS deficit
$\delta$ is governed by comparing the source entropy rate $h$ with
$\Phi_q(\delta)=2\mathsf H(\delta)+\delta\log(2q)$, where $\mathsf H$ is
binary entropy and logarithms are natural.  This is an explicit upper bound on the
exponential rate of a high-LCS deletion ball; the theorem adds
finite-scale and local-dependency terms.
\end{theorem}

The formal transfer theorem and its scalar sufficient conditions appear in
\cref{thm:entropy-transfer}.  A uniform occurrence-wise square-free
substitution with block length $m$ and $K$ choices supplies entropy rate
$(\log K)/m$ and boundary loss $2\log K$, so the general substitution
criterion follows immediately.  For the explicit ternary family, the same
analysis also gives exponentially many successful choices and a randomized
construction.

The transfer principle also applies to extremal square-free words.
A square-free word is \emph{extremal} if inserting any letter at any position
creates a square.  Here the square appears in the modified word, whereas
synchronization concerns adjacent intervals of the original word.  To obtain
both properties, we need occurrence-wise choices that preserve extremality
and retain enough local entropy to exclude approximate repetitions.  The
multivalued substitution used to construct extremal words provides such choices.

\begin{theorem}[Extremal square-free synchronization strings]
\label{thm:extremal-sync-intro}
For every integer $n\ge2493$, there is a ternary word of length $n$ that is
simultaneously extremal square-free and a
$200001/200002$-synchronization string.
\end{theorem}

The input here is the Thue-digraph substitution of Mol and Rampersad, whose
two occurrence-wise images have lengths $41$ and $52$.  Pairing consecutive
occurrences as $QR$ or $RQ$ exposes one binary choice on a fixed
length-$93$ block while preserving extremality.  A bounded-imbalance length
decomposition reaches every $n\ge2493$, and the transfer theorem converts
the resulting local entropy into synchronization.  Pairing is what makes the
transfer possible: although individual images have different lengths, the
paired alternatives have the same length, so changing one choice does not
move the later block boundaries.  Thus a substitution designed to prevent
every square-free insertion still supplies the local freedom needed for
synchronization.  Once these equal-length choices have been identified, the
transfer theorem requires only their entropy bound and square-free support.

For the original alphabet question, a larger Brinkhuis family gives a
stronger explicit parameter.

\begin{theorem}[Computer-assisted quantitative strengthening]
\label{thm:quantitative-intro}
For every integer $n\ge1$, there exists a word $S\in\ternary^n$ such that,
for all $1\le i<j<k\le n+1$,
\begin{equation}\label{eq:quantitative-gap}
  \ED(S[i,j),S[j,k))\ge\mainboundarygap(k-i).
\end{equation}
Consequently, for every $\eps>\mainthreshold$, ternary
$\eps$-synchronization strings exist for every length.  In particular, one
may take $\eps=\maineps$.
\end{theorem}

By \cref{cor:infinite-ternary}, the same inequality holds for an infinite
ternary word.  This refinement uses a 54-uniform family with 952 choices per
outer letter, supported deletion-ball bounds, and finite checks of matching
runs and bounded deletion shells; the
six-codeword proof of \cref{thm:ternary-intro} is independent of these checks.
The larger codebook and exact-arithmetic
verifiers for the finite shell checks and local-lemma calculations are
provided with the manuscript; the public distribution is the
\href{https://github.com/xxxhhbb/ternary-synchronization-strings-artifact}
{GitHub repository}.

In particular, $\eps=216/217\approx0.995392$ improves on the explicit
four-symbol parameter $0.999606$ obtained by Bhattacharya et
al.~\cite[Corollary~1.10]{BhattacharyaEtAl2026}.  The direct
predecessor for the alphabet question is the qualitative four-symbol theorem
of Cheng et al.~\cite{ChengEtAl2019}.  Here a smaller $\eps$ means a larger
guaranteed normalized edit distance $1-\eps$.  Thus the refinement improves
both the alphabet size and this explicit parameter, although the optimal
ternary parameter remains open.

The method also extends to synchronization circles, which must remain
synchronization strings after every
cyclic rotation.  Cheng et al.~\cite[Theorem~3.4]{ChengEtAl2019} constructed such circles
of every length over alphabets of size $O(\eps^{-2})$ by joining two
synchronization strings over disjoint alphabets.  We address the extremal
alphabet question for unbounded circle lengths.  The occurrence-wise
argument extends to this setting: a
circularly square-free outer word produces a circularly square-free output,
and the same conditional-entropy local lemma controls the long cyclic arcs.
Since ternary circularly square-free words exist at every length at least
$18$, we obtain a ternary $2001/2002$-synchronization circle of every length
$18N$ with $N\ge18$.  Hence three is also the exact minimum alphabet size
supporting synchronization circles of unbounded length.  Appendix
\ref{app:all-length-circles} strengthens this by a computer-assisted
argument: every length $n\ge10008$ admits a ternary
$5201/5202$-synchronization circle.

Finally, the local-lemma analysis gives quantitative bounds on the number of
valid words and the time needed to find one.
\begin{theorem}[Computer-assisted quantitative consequences]\label{thm:extensions-intro}
For every $n\ge504$, at least $\exp(n/9)$ words in $\ternary^n$ satisfy
\cref{eq:quantitative-gap}.  Moreover, a word satisfying
\cref{eq:quantitative-gap} can be generated by a Las Vegas randomized algorithm
in expected time $O(n^2\log^3(n+2))$ and space $O(n^2)$ machine words.
\end{theorem}

The counting and randomized assertions are proved as
\cref{cor:explicit-count,thm:long-distance-construction} in
\cref{sec:extensions}.  The product form of the local lemma yields the
count.  For the algorithm, we additionally separate every disjoint interval
pair of total length at least $30000\log n$ by the same normalized gap.
This long-distance property allows all longer violations to be reduced to
events on $O(\log n)$ symbols.  A cache indexed by the two interval starts
then avoids an exhaustive rescan after each Moser--Tardos resampling.

\subsection{Proof overview}

We describe the six-codeword proof of \cref{thm:ternary-intro}.  It has three
parts: square-freeness handles bounded scales, conditional entropy and the
local lemma handle long scales, and a deterministic reduction passes from
equal adjacent blocks to arbitrary adjacent intervals.

\paragraph{Short scales and the random source.}
Suppose $X$ and $Y$ are adjacent intervals of the same length $\ell$.  The
standard identity between edit distance and longest common subsequence gives
$\ED(X,Y)=2(\ell-\LCS(X,Y))$.  Thus a fixed relative LCS deficit is exactly a fixed normalized edit-distance
gap for equal blocks.  It is therefore enough at first to rule out events of
the form
\begin{equation}\label{eq:intro-bad-event}
  \LCS\bigl(S[t,t+\ell),S[t+\ell,t+2\ell)\bigr)
     >(1-\eta)\ell.
\end{equation}
There are quadratically many such events and their variable sets overlap
heavily, so a union bound is not adequate.  The local lemma will be used only
after we obtain an exponentially small probability for each event and keep
track of its scale.

Equal blocks also expose the role of square-freeness.  If $X=Y$, then $XY$
is a square.  Hence in a square-free word $X\ne Y$, and therefore
$\LCS(X,Y)\le \ell-1$.  For $\ell\le L$, this is already the relative deficit
$\LCS(X,Y)\le(1-1/L)\ell$.  Square-freeness therefore settles all scales up
to a cutoff $L$ without probability or enumeration.  What remains is to
replace the vanishing deficit $1/\ell$ by a constant deficit when
$\ell>L$.

To obtain the random source, fix a ternary square-free outer word
$T=T_1\cdots T_N$.  For each occurrence
$T_i$, choose a bit $R_i$ independently and replace that occurrence by the
length-18 word $B_{T_i,R_i}$ from the corrected Brinkhuis family.  The
resulting word is
$W(T,R)=B_{T_1,R_1}B_{T_2,R_2}\cdots B_{T_N,R_N}$.
The Brinkhuis guarantee holds for every outer square-free word and every
occurrence-wise choice vector.  Consequently
the support of our probability space consists entirely of ternary
square-free words.  We may use the choices $R_i$ to attack long approximate
repetitions without ever reintroducing an exact short repetition.

Two occurrences of the same outer letter may therefore use different
images.  The bits $R_i$ select among ternary blocks, as illustrated in
\cref{fig:intro-construction}; the output still uses only three symbols.

\paragraph{Local entropy bounds the probability of a long match.}
Fix the event in \cref{eq:intro-bad-event}, and write its two intervals as
$X$ and $Y$.  Reveal every $R_i$ except those whose complete macroblocks lie
inside $X$.  Then $Y$ and the boundary blocks are fixed, while every
unrevealed macroblock lies entirely inside $X$.  An interval of length $\ell$ contains at least
$\ell/18-2$ complete macroblocks, so that many independent fair bits remain.

The two images available for a fixed outer letter are distinct.  Once the
macroblock boundaries and $T$ are fixed, different remaining choice vectors
therefore give different values of the corresponding portion of $X$.
Conditioned on everything already revealed, every exact candidate value of
$X$ has probability at most
\begin{equation}\label{eq:intro-point-mass}
  2^{-(\ell/18-2)}.
\end{equation}
This is a local min-entropy statement: it remains valid no matter how the
word outside $X$ has been fixed.  That robustness under conditioning is what
allows overlapping interval events to be handled later.

It remains to count how many length-$\ell$ ternary words can have LCS more than
$(1-\eta)\ell$ with the now fixed word $Y$.  If the LCS deficit is $d$, one
may specify such a word by choosing the $d$ positions deleted from $Y$, the
$d$ positions deleted from the candidate $X$, and the ternary symbols placed
at the latter positions.  This standard deletion-ball encoding gives the
upper bound
\begin{equation}\label{eq:intro-deletion-ball}
  \sum_{0\le d<\eta\ell}
     \binom{\ell}{d}^{\!2}3^d
  =\exp\!\bigl(O(\eta\log(1/\eta))\ell\bigr).
\end{equation}
Multiplying \cref{eq:intro-deletion-ball} by the point-mass bound
\cref{eq:intro-point-mass} shows that the fixed adjacent pair is bad with
probability at most $\exp(-c\ell)$ for an absolute $c>0$ once $\eta$ is
sufficiently small.  For the six-codeword
family we take $\eta=1/1000$; elementary estimates verify the required
numerical inequality.  Its content is simple: the deletion ball has
entropy rate tending to zero with $\eta$, whereas the
occurrence-wise Brinkhuis choices supply a fixed positive entropy rate.

Ignoring the two boundary macroblocks, the six-codeword family supplies
$(\log 2)/18\approx0.0385$ nats of choice entropy per output character.  The
leading entropy term in the deletion-ball bound is
$2\mathsf H(\eta)+\eta\log 3$, where
$\mathsf H(\eta)=-\eta\log\eta-(1-\eta)\log(1-\eta)$; it is about $0.0169$ at
$\eta=1/1000$.  This strict entropy margin drives
the long-scale argument.  The formal proof also retains the boundary losses,
integer roundings, and local-lemma charge.

\paragraph{Enforcing all scales.}
The interval events are not independent, but their dependencies are
geometric.  An event at scale $\ell$ is determined only by macroblock choices
meeting a character interval of length $2\ell$.  At another scale $s$, only
$O(\ell+s)$ starting positions can overlap it.  Assigning an exponentially
decaying charge to a scale-$s$ event makes the total charge of all its
neighbors summable over $s$.

An asymmetric Lov\'asz Local Lemma argument in the variable model formalizes
this calculation.
The proof uses the actual scale of each neighbor rather than bounding all
events by the smallest probability.  The exponential event tail pays for
the linear number of positions at every scale, and the product of the
remaining local-lemma factors stays large enough.  We conclude that there is
one occurrence-wise choice vector for which no long event
\cref{eq:intro-bad-event} occurs.  Together with square-freeness on
$\ell\le1000$, this gives the equal-block LCS gap at every scale.

We now pass to arbitrary adjacent intervals, whose lengths need not
be equal.  Let the shorter length be $m$ and let the length imbalance be
$z$.  The equal-block estimate on the length-$m$ pieces adjacent to the
common boundary gives an edit-distance lower bound $2\eta m-z$ after paying
for the unmatched excess; the difference in lengths independently gives the
lower bound $z$.  Combining them gives
$\ED(A,B)\ge\max\{z,2\eta m-z\}$.
Dividing by $|A|+|B|=2m+z$ and considering whether $z$ is smaller or larger
than $\eta m$ gives normalized edit distance at least
$\eta/(2+\eta)$.  With $\eta=1/1000$ this is $1/2001$.  Taking the slightly
smaller strict gap $1/2002$ proves that every length has a ternary
$2001/2002$-synchronization string.  Combining this with the known binary
obstruction closes the alphabet question.

\paragraph{Extensions of the argument.}
Abstracting the preceding calculation gives the local-entropy transfer
theorem.  Its main condition is that the exponential rate of a high-LCS
deletion ball be smaller than the interval conditional min-entropy rate.
Lower-order terms account for block boundaries and overlapping dependencies.
The theorem therefore applies whenever a square-free source supplies the
required local entropy, regardless of the combinatorial origin of that
source.

For the stronger constant, we use the length-54 family with 952
choices per outer-letter occurrence.  Its larger conditional entropy permits
a larger long-scale LCS deficit.  At bounded scales, however,
square-freeness by itself certifies only the first deletion
shell.  We therefore verify the remaining one- and two-deletion obstructions
for the finite codebook and analyze unequal adjacent intervals directly,
avoiding the loss in the generic equal-to-unequal conversion.  Supported
fillings and phase-dependent bounds on uninterrupted matching runs further
reduce the deletion-ball volume.  This yields
the non-strict normalized edit-distance lower bound $1/216$, and hence every
parameter
$\eps>215/216$; in particular, one may take $216/217$.
The quantitative argument is summarized in
\cref{sec:quantitative-summary} and proved in
Appendix~\ref{sec:quantitative-strengthening}.

On a circularly square-free outer word, the same conditioning and
scale-sensitive local lemma apply to cyclic arcs, giving the
$18N$-length circle result.  To reach every sufficiently large length,
Appendix~\ref{app:all-length-circles} deletes one short factor from
the uniform output, proves that every new square is confined to a bounded
neighborhood of the resulting seam, and verifies the finitely many boundary
configurations.

The product form of the local lemma lower-bounds the measure of successful
choice vectors; injectivity then converts this bound into exponentially many
ternary words.  The same bad events can be detected by LCS dynamic
programming, and Moser--Tardos resamples the macroblocks of a violated event.
The charge calculation bounds the expected number of resamplings, giving the
Las Vegas polynomial-time construction.

\subsection{Related work}

\paragraph{Synchronization strings.}
Haeupler and Shahrasbi introduced synchronization strings for approaching the
Singleton bound under insertions and deletions, and subsequent work developed
explicit and locally decodable constructions and applications to channel
simulation and indexing
\cite{HaeuplerShahrasbi2017,HaeuplerShahrasbi2018,
HaeuplerShahrasbiVitercik2018,HaeuplerRubinsteinShahrasbi2019}.
Cheng, Haeupler, Li, Shahrasbi, and Wu
\cite[Section~1.3.2 and Theorem~6.4]{ChengEtAl2019} proved that four
symbols suffice for some fixed
$\eps<1$, two symbols do not, and the ternary case was left open.  Their
prescribed-parameter bounds concern the different regime in which the
alphabet grows as $\eps\to0$.  Bhattacharya, Dey, Goldenberg, Habib,
Haeupler, Karthik C.~S., and
Kouck\'y~\cite[Corollary~1.10]{BhattacharyaEtAl2026} obtained the
explicit four-symbol value $0.999606$ as a corollary of their Hamming-to-edit
embedding construction.  Since that construction is quaternary, the ternary
extremal case remained open.

\paragraph{Pattern avoidance.}
Square-free words originate with Thue~\cite{Thue1906}.  Brinkhuis-type
substitutions were developed to prove exponential abundance of ternary
square-free words
\cite{Brinkhuis1983,EkhadZeilberger1998,EkhadZeilbergerErratum2001,
Grimm2001,SollamiDouglasLiebmann2016}.  We use the corrected occurrence-wise
families from this literature, but reinterpret their choices as conditional
anti-concentration rather than counting.  Avoidance under LCS similarity was
also studied by Camungol and Rampersad, who gave an existence bound for
sufficiently large alphabets and computational bounds beginning at three
symbols~\cite[Theorem~1 and Tables~1--2]{CamungolRampersad2016}.
Grytczuk, Kordulewski, and Niewiadomski introduced extremal square-free words,
and Mol and Rampersad determined their attainable ternary lengths using a
multivalued Thue-digraph substitution
\cite{GrytczukKordulewskiNiewiadomski2020,MolRampersad2021}.  Our second
application uses the occurrence-wise choices in that construction as a new
block-local source.

\paragraph{Probabilistic tools.}
The asymmetric local lemma and the Moser--Tardos resampling theorem are used
in their standard forms~\cite{ErdosLovasz1975,MoserTardos2010}.  Beck's
local-lemma construction of binary sequences with separated repeated
intervals is conceptually related~\cite{Beck1984}, but it controls a different
distance property.  The small-alphabet construction of Cheng et
al.~\cite[Sections~3--4]{ChengEtAl2019} also combines a nonuniform random source, deletion
counting, and the local lemma.  The source condition here is different:
bounded-range independent variables retain interval conditional min-entropy,
and every outcome remains square-free.  The local lemma therefore runs
entirely inside the square-free support.

\paragraph{Organization.}
\Cref{sec:preliminaries} reviews the standard distance and local-lemma tools,
and \cref{sec:brinkhuis} introduces the classical Brinkhuis input.
\Cref{sec:ternary-proof} proves the alphabet theorem, and
\cref{sec:structural} states the local-entropy transfer principle and develops
its substitution and cyclic consequences, together with an overview of the
extremal application.  \Cref{sec:quantitative-summary} presents
the stronger explicit parameter, while abundance and the randomized
construction appear in \cref{sec:extensions}.  The appendices contain the
general transfer proof, the extremal square-free application, the all-length
circle refinement, the quantitative argument, and the finite
calculations underlying the computer-assisted refinements.

%% file: sections/02_preliminaries.tex
\section{Preliminaries}\label{sec:preliminaries}

All logarithms are natural.  A word $S=S_1\cdots S_n$ has length $|S|=n$,
and $S[i,j)=S_iS_{i+1}\cdots S_{j-1}$ for $1\le i\le j\le n+1$.
Intervals are always nonempty unless stated otherwise.

\subsection{Edit distance, common subsequences, and synchronization}

The insertion--deletion distance $\ED(x,y)$ is the minimum number of symbol
insertions and deletions transforming $x$ into $y$.  The longest common
subsequence length is denoted by $\LCS(x,y)$.  They satisfy the standard identity
\begin{equation}\label{eq:ed-lcs}
\ED(x,y)=|x|+|y|-2\LCS(x,y).
\end{equation}
It follows by deleting to a common subsequence and then inserting the missing
symbols; see, for example, \cite[Definitions~2.2--2.3]{ChengEtAl2019}.

\begin{definition}[Square-free word]\label{def:square-free}
A \emph{square} is a nonempty word repeated twice consecutively, $XX$.  A
word is \emph{square-free} if none of its factors is a square.
\end{definition}

By \cref{def:sync}, every synchronization string with $\eps<1$ is
square-free.  Both square-freeness and every fixed synchronization inequality
are inherited by factors: adjacent intervals inside a factor are unchanged adjacent
intervals of the original word.  In particular, both properties are inherited
by prefixes.  The next lemma reduces the synchronization condition to an
equal-length LCS gap.

\begin{lemma}[From equal blocks to arbitrary intervals]\label{lem:equal-to-full}
Let $0<\eta<1$.  Suppose a word $S$ satisfies
\begin{equation}\label{eq:equal-gap-general}
  \LCS(S[t,t+\ell),S[t+\ell,t+2\ell))
  \le(1-\eta)\ell
\end{equation}
for every available start $t$ and length $\ell\ge1$.  Then $S$ is an
$\eps$-synchronization string whenever
$1-\eps<\eta/(2+\eta)$.  In particular, it has synchronization gap
$\eta/3$.
\end{lemma}

\begin{proof}
Let $A,B$ be arbitrary adjacent nonempty intervals, put
$m=\min\{|A|,|B|\}$ and $z=\lvert |A|-|B|\rvert$, and let $A_0$ and $B_0$
be respectively the length-$m$ suffix of $A$ and prefix of $B$.  These are
adjacent equal-length intervals, so
$\ED(A_0,B_0)\ge 2\eta m$ by
\cref{eq:ed-lcs,eq:equal-gap-general}.  Deleting the excess $z$ symbols from
the longer interval and using the triangle inequality gives
$\ED(A,B)\ge 2\eta m-z$.  The difference in lengths also gives
$\ED(A,B)\ge z$; hence
\[
  \frac{\ED(A,B)}{|A|+|B|}
  \ge \frac{\max\{z,2\eta m-z\}}{2m+z}.
\]
If $z\ge\eta m$, the right-hand side is at least
$z/(2m+z)\ge\eta/(2+\eta)$.  If $z\le\eta m$, it is at least
$(2\eta m-z)/(2m+z)\ge\eta/(2+\eta)$.  Thus every adjacent pair has
normalized edit distance at least $\eta/(2+\eta)$.  Any strictly smaller
gap therefore satisfies \cref{def:sync}; in particular, $\eta/3$ does because
$\eta<1$.
\end{proof}

\subsection{Local-lemma tools}

Let $\mathcal V$ be a finite family of mutually independent random variables,
and let $\cE$ be a finite family of bad events.  For each $E\in\cE$, let
$\vbl(E)\subseteq\mathcal V$ be a set of variables determining $E$, and let
$\Gamma(E)$ consist of the events $F\ne E$ for which
$\vbl(E)\cap\vbl(F)\ne\varnothing$.  We use the following standard forms of
the local lemma and its algorithmic
counterpart~\cite{ErdosLovasz1975,MoserTardos2010}, including the product
bound used later in the abundance argument.

\begin{theorem}[Asymmetric Lov\'asz Local Lemma]\label{thm:lll}
Suppose numbers $x(E)\in(0,1)$ satisfy
\begin{equation}\label{eq:lll-criterion}
  \Pr[E]\le x(E)\prod_{F\in\Gamma(E)}(1-x(F))
  \qquad\text{for every }E\in\cE.
\end{equation}
Then
\begin{equation}\label{eq:lll-product}
  \Pr\!\left[\bigcap_{E\in\cE}\overline E\right]
  \ge\prod_{E\in\cE}(1-x(E))>0.
\end{equation}
\end{theorem}

\begin{theorem}[Moser--Tardos]\label{thm:moser-tardos}
Under the same hypotheses and \cref{eq:lll-criterion}, the algorithm that
repeatedly chooses any currently true event $E$ and resamples the variables in
$\vbl(E)$ terminates almost surely.  The expected number of times it resamples
$E$ is at most $x(E)/(1-x(E))$.
\end{theorem}

%% file: sections/03_brinkhuis_source.tex
\section{The square-free source: Brinkhuis substitutions}\label{sec:brinkhuis}

Our construction starts from a classical square-free substitution in its
\emph{occurrence-wise} form: different occurrences of the same source letter
may choose different images.

\begin{definition}[Uniform Brinkhuis family]\label{def:brinkhuis}
An $m$-Brinkhuis $K$-triple consists of three sets
$\cC_0,\cC_1,\cC_2\subseteq\ternary^m$, each of size $K$, such that for every
ternary square-free word $abc$ of length three and every independent choice
$U\in\cC_a$, $V\in\cC_b$, $Z\in\cC_c$, the concatenation $UVZ$ is
square-free.  The choices of $U$ and $Z$ are independent even when $a=c$.
\end{definition}

The length-three condition corrects the original, insufficient two-block
formulation~\cite{EkhadZeilbergerErratum2001}; see also
\cite[Definition~1]{Grimm2001} and
\cite[Definition~1]{SollamiDouglasLiebmann2016}.
We use the following six length-$18$ words of Ekhad and Zeilberger:
\begin{equation}\label{eq:codebook}
\begin{aligned}
B_{0,0}&=210201202120102012,&
B_{0,1}&=210201021202102012,\\
B_{1,0}&=021012010201210120,&
B_{1,1}&=021012102010210120,\\
B_{2,0}&=102120121012021201,&
B_{2,1}&=102120210121021201.
\end{aligned}
\end{equation}

\begin{proposition}[Classical 18-uniform Brinkhuis family]
\label{prop:codebook}
The six words in \cref{eq:codebook} form an 18-Brinkhuis 2-triple.  More
generally, if $T=T_1\cdots T_N$ is any ternary square-free word and
$R\in\bits^N$, then
\begin{equation}\label{eq:random-output}
  W(T,R)=B_{T_1,R_1}B_{T_2,R_2}\cdots B_{T_N,R_N}
\end{equation}
is square-free.  The map $R\mapsto W(T,R)$ is injective.
\end{proposition}

\begin{proof}
The Brinkhuis assertion, including independent choices at different
occurrences, is the corrected 18-uniform construction of
Ekhad and Zeilberger and its standard generalized-Brinkhuis interpretation
\cite{EkhadZeilberger1998,EkhadZeilbergerErratum2001}; see also
\cite[Definition~1 and the subsequent discussion]{Grimm2001} and
\cite[Definition~1]{SollamiDouglasLiebmann2016}.  For injectivity, the two displayed images of
each fixed outer letter are
distinct.  Since the macroblock boundaries and $T$ are fixed, the $i$th
macroblock determines $R_i$.
\end{proof}

For every $N$ there is a ternary square-free word of length $N$; for example,
take a prefix of Thue's classical infinite word~\cite{Thue1906}.  For the
rest of the explicit proof, fix such a word $T$ and take
$R_1,\ldots,R_N$ independently and uniformly from $\bits$.  Every outcome is
square-free by \cref{prop:codebook}.

The cited works use the $2^N$ occurrence-wise choices to count square-free
words.  We use them to bound the probability of a long common subsequence:
after the choices of all macroblocks not contained in a character interval
are fixed, linearly many independent choices remain inside it.  The next
section turns this conditional entropy into an LCS gap at every scale.

%% file: sections/04_ternary_proof.tex
\section{Ternary synchronization strings}\label{sec:ternary-proof}

The six codewords in \cref{eq:codebook} suffice to settle the alphabet
question.

\begin{theorem}[The six-codeword construction]
\label{thm:small-family}
For every integer $n\ge1$, there is a ternary $\baseeps$-synchronization
string of length $n$; equivalently, one may take synchronization gap
$\basegap$.
\end{theorem}

Throughout this section, set $\eta=\baseeta$ and $L=\basecutoff$.
Fix an outer square-free word $T$ and use the random output $W=W(T,R)$ from
\cref{eq:random-output}.  At every long scale, the argument bounds a bad-event
probability by the product of a conditional point-mass bound and the volume
of an LCS deletion ball.  The former decays faster than the latter grows.
A scale-sensitive
local-lemma calculation removes all long violations; square-freeness covers
the complementary short scales.

\subsection{Conditional entropy and high-LCS events}

We first record the standard LCS-witness count from probabilistic
synchronization-string constructions
\cite{HaeuplerShahrasbi2017,ChengEtAl2019}, retaining the exponent needed
here.  The source-specific step is the conditional point-mass bound in
\cref{lem:conditional-mass}, which remains valid after all choices outside
the interval have been exposed.

\begin{lemma}[Standard high-LCS deletion-ball bound]\label{lem:lcs-ball}
Fix $y\in\alphabet^\ell$, where $|\alphabet|=q$, and let $0<\eta<1$.
Then
\begin{equation}\label{eq:lcs-ball}
\left|\left\{x\in\alphabet^\ell:
\LCS(x,y)>(1-\eta)\ell\right\}\right|
\le
\sum_{0\le d<\eta\ell}\binom{\ell}{d}^2q^d.
\end{equation}
\end{lemma}

\begin{proof}
If $\LCS(x,y)=\ell-d$, fix one witnessing common subsequence.  It is
described by the $d$ positions deleted from $y$, the $d$ positions deleted
from $x$, and the arbitrary symbols in the remaining $d$ positions of $x$.
There are at most $\binom{\ell}{d}^2q^d$ such descriptions.  A word may have
several descriptions, which only makes this an upper bound.  Finally,
$\LCS(x,y)>(1-\eta)\ell$ means $d<\eta\ell$; summing over these integers
proves \cref{eq:lcs-ball}.
\end{proof}

\begin{lemma}[Explicit ternary volume]\label{lem:ternary-volume}
For the parameters fixed above, every $\ell\ge L$ and every
$y\in\ternary^\ell$ satisfy
\begin{equation}\label{eq:ternary-volume}
  \left|\left\{x\in\ternary^\ell:
  \LCS(x,y)>(1-\eta)\ell\right\}\right|
  \le \exp(9\ell/500).
\end{equation}
\end{lemma}

\begin{proof}
For $1\le d<\eta\ell$, \cref{lem:binomial-estimates}(i)--(ii) gives
$\binom{\ell}{d}\le\exp(\eta\ell\log(e/\eta))$, while the $d=0$
summand in \cref{eq:lcs-ball} equals $1$.  Since $\eta\ell\ge1$, there
are at most $2^{\eta\ell}$ admissible integers $d$ by
\cref{lem:binomial-estimates}(iii), and $3^d\le3^{\eta\ell}$.  Hence the
right-hand side of \cref{eq:lcs-ball} is at most
$\exp(\eta\ell(\log2+2\log(e/\eta)+\log3))$, which is strictly smaller
than $\exp(9\ell/500)$ by \cref{lem:volume-arithmetic}.
\end{proof}

We now bound the probability that the random output lies in such a deletion
ball.  For an available character start $t$ and an integer $\ell>L$, define
\begin{equation}\label{eq:bad-event}
  E_{t,\ell}=
  \left\{
  \LCS\bigl(W[t,t+\ell),W[t+\ell,t+2\ell)\bigr)
  >(1-\eta)\ell
  \right\}.
\end{equation}

\begin{lemma}[Whole macroblocks]\label{lem:whole-blocks}
Every character interval of length $\ell$ contains at least
$\ell/18-2$ complete Brinkhuis macroblocks.
\end{lemma}

\begin{proof}
The interval intersects at most two macroblocks without containing them: one
at each endpoint.  If it contains $s$ complete macroblocks, its characters
are covered by these $s$ blocks together with at most two further length-18
blocks.  Hence $\ell\le18(s+2)$, or $s\ge\ell/18-2$.
\end{proof}

\begin{lemma}[Conditional point mass]\label{lem:conditional-mass}
Write the two length-$\ell$ intervals in \cref{eq:bad-event} as $X,Y$.
There are $s\ge\ell/18-2$ choice variables such that,
after conditioning on all other choices, $Y$ is fixed and every exact value of
$X$ has conditional probability at most $2^{-s}$.
\end{lemma}

\begin{proof}
Take the choice variables of all complete macroblocks contained in $X$ and
condition on every other $R_i$.  None of the unconditioned macroblocks meets
$Y$, so this conditioning fixes $Y$, including the choice of a macroblock
that might cross the boundary between $X$ and $Y$.  The number $s$ of
unconditioned variables satisfies the bound in \cref{lem:whole-blocks}.

Each unconditioned variable occupies a known full macroblock inside $X$.  For
its fixed outer symbol, the two possible codewords in
\cref{eq:codebook} are distinct.  Therefore an exact value of $X$ determines
all $s$ unconditioned bits.  Those bits are independent and uniform, so any
one value of $X$ has conditional probability at most $2^{-s}$.
\end{proof}

\begin{lemma}[Exponential event tail]\label{lem:event-tail}
For every event in \cref{eq:bad-event},
\begin{equation}\label{eq:event-tail}
  \Pr[E_{t,\ell}]<e^{-9\ell/500}.
\end{equation}
\end{lemma}

\begin{proof}
Let $\mathcal F$ be the sigma-field generated by the choices conditioned on
in \cref{lem:conditional-mass}, and fix an atom $f$ of $\mathcal F$.  For the
resulting fixed value of $Y$, \cref{lem:ternary-volume} bounds the number of
bad values of $X$ by $e^{9\ell/500}$.  Each has conditional probability at
most $2^{-s}$.  Using $s\ge\ell/18-2$ gives
\begin{equation}\label{eq:event-intermediate}
  \Pr[E_{t,\ell}\mid\mathcal F=f]
  \le4\exp\!\left[-\left(\frac{\log2}{18}-\frac9{500}\right)\ell\right].
\end{equation}
This estimate holds for every atom $f$, so the law of total probability
gives the same unconditional bound.  The full
factor $4$ is absorbed uniformly for every $\ell>1000$ by
\cref{lem:event-arithmetic}, which turns
\cref{eq:event-intermediate} into \cref{eq:event-tail}.
\end{proof}

\subsection{Eliminating every long bad event}

An event $E_{t,\ell}$ is measurable with respect to the choice variables of
the length-18 macroblocks intersecting its character support
$[t,t+2\ell)$.  Join two events when these variable sets intersect.  This is
the variable-dependency graph used in \cref{thm:lll}.  Events whose declared
variable sets are disjoint are functions of disjoint families of independent
occurrence choices, so this graph has the required dependency property.

\begin{lemma}[Neighbors at one scale]\label{lem:neighbor-count}
For a fixed event of length $\ell$, the number of neighboring events of
length $r$ is at most $2\ell+2r+37\le3(\ell+r)$.
\end{lemma}

\begin{proof}
Let the fixed support start at $t$ and the other support start at $u$.  If
their variable sets share a macroblock, then the two character supports are
separated by fewer than 18 characters.  Consequently,
$t-2r-18<u<t+2\ell+18$, an open interval containing at most
$2\ell+2r+37$ integer starts. Counting all of them, including starts that
would place the event outside the word or reproduce the fixed event itself,
only overestimates the number of neighbors. Because $\ell,r>L$, we have
$37\le\ell+r$, proving the stated second inequality.
\end{proof}

Assign the local-lemma charge
\begin{equation}\label{eq:explicit-charge}
  x(E_{t,\ell})=e^{-\ell/80}.
\end{equation}

\begin{lemma}[Neighbor charge]\label{lem:neighbor-charge}
For every event $E_{t,\ell}$,
\begin{equation}\label{eq:neighbor-charge}
  \sum_{F\sim E_{t,\ell}}x(F)<\frac{\ell}{490}.
\end{equation}
\end{lemma}

\begin{proof}
The actual neighboring events have integral lengths $r>L$.  By
\cref{lem:neighbor-count,eq:explicit-charge}, adding the nonnegative term at
$r=L$ bounds the sum by
$3\sum_{r\ge L}(\ell+r)e^{-r/80}$.  The geometric-series evaluation and
strict numerical comparison in \cref{lem:neighbor-arithmetic} show that this
quantity is smaller than $\ell/490$.
\end{proof}

\begin{proposition}[Long equal blocks]\label{prop:long-equal-blocks}
There exists a choice vector $R$ for which no event
$E_{t,\ell}$ with $\ell>L$ occurs.
\end{proposition}

\begin{proof}
Every charge in \cref{eq:explicit-charge} is smaller than
$e^{-25/2}<1/2$.  Applying \cref{lem:small-charge-product} and then
\cref{eq:neighbor-charge} gives
$\prod_{F\sim E_{t,\ell}}(1-x(F))>e^{-\ell/245}$.  Since
$9/500>1/80+1/245$ by \cref{lem:final-margins}, the event bound gives
\[
  \Pr[E_{t,\ell}]
  < e^{-9\ell/500}
  < x(E_{t,\ell})\prod_{F\sim E_{t,\ell}}(1-x(F)),
\]
which is the asymmetric criterion
\cref{eq:lll-criterion}.  The probability space and event family are finite
for the fixed outer word, so \cref{thm:lll} supplies an outcome avoiding all
long bad events.
\end{proof}

\subsection{All scales and all adjacent intervals}

\begin{proposition}[Uniform equal-length gap]\label{prop:all-equal-blocks}
For every output length $n$, there is a ternary word $S\in\ternary^n$ such
that every available adjacent equal-length pair $X,Y$ of common length
$\ell\ge1$ satisfies
\begin{equation}\label{eq:all-equal-gap}
  \LCS(X,Y)\le(1-\eta)\ell.
\end{equation}
\end{proposition}

\begin{proof}
Choose an outer word of length $N=\lceil n/18\rceil$ and use the outcome from
\cref{prop:long-equal-blocks}.  For $\ell>L$, the absence of the strict bad
event \cref{eq:bad-event} gives \cref{eq:all-equal-gap}.  For
$1\le\ell\le L$, \cref{prop:codebook} makes the full output square-free.
Thus $X\ne Y$, because equality would create the square $XX$.  A common
subsequence of length $\ell$ would use every position of both words and force
equality, so $\LCS(X,Y)\le\ell-1\le(1-1/L)\ell=(1-\eta)\ell$, where the
second inequality uses $\ell\le L$, including the endpoint $\ell=L$.
The full output has length $18N\ge n$;
take its length-$n$ prefix.  Every adjacent equal-length pair inside that
prefix is the identical pair inside the full output, so
\cref{eq:all-equal-gap} is inherited directly.
\end{proof}

\begin{proof}[Proof of \cref{thm:small-family}]
Apply \cref{lem:equal-to-full} to the word in
\cref{prop:all-equal-blocks}.  Its normalized edit distance between arbitrary
adjacent intervals is at least $\eta/(2+\eta)=1/2001>\basegap$.  Thus the
strict synchronization inequality holds with gap $\basegap$, or
equivalently with parameter $\baseeps$.
\end{proof}

\begin{proof}[Proof of \cref{thm:ternary-intro}]
The existence assertion is \cref{thm:small-family}.  The binary lower bound
in~\cite[Section~1.3.2]{ChengEtAl2019} proves that no fixed $\eps<1$ permits arbitrarily
long binary synchronization strings.  Hence the ternary alphabet is
optimal.
\end{proof}

%% file: sections/05_structural_theorem.tex
\section{From local entropy to synchronization}\label{sec:structural}

The preceding proof needs two properties of the source: every outcome is
square-free, and every long interval retains linear conditional min-entropy
after the variables outside it are fixed.  We now state a transfer theorem
in terms of these properties and apply it to uniform substitutions, extremal
square-free words, and synchronization circles.

Let $\mathsf H(u)=-u\log u-(1-u)\log(1-u)$ denote binary entropy, and define
the following upper bound on the deletion-ball entropy rate:
\begin{equation}\label{eq:transfer-volume-rate}
  \Phi_q(\delta)=2\mathsf H(\delta)+\delta\log(2q),
  \qquad 0<\delta\le\frac12.
\end{equation}

\begin{definition}[Block-local square-free source]
\label{def:block-local-source}
Let $Q$ be an alphabet of size $q$.  An $(m,h,b)$-block-local square-free
source is a random word $W\in Q^M$ together with a partition of its positions
into consecutive blocks of length at most $m$ and one independent finite-valued
random variable per block, such that the contents of a block depend only on its own
variable, every word in the support of $W$ is square-free, and the following
conditional point-mass bound holds.

For every character interval $I$ of length $\ell$, leave unconditioned the
variables of the blocks contained completely in $I$ and condition on all
remaining variables.  Under every conditioning of positive probability,
\begin{equation}\label{eq:block-local-min-entropy}
  \max_{x\in Q^\ell}
  \Pr[W[I]=x\mid\text{the conditioning}]
  \le \exp(-h\ell+b).
\end{equation}
Here $h>0$ is the local entropy rate and $b\ge0$ is a boundary loss.
\end{definition}

For integers $L\ge1$ and real numbers $a>0$, define the dependency tail
\begin{equation}\label{eq:transfer-dependency-tail}
  D_L(a):=3\sum_{r>L}\left(1+\frac rL\right)e^{-ar}.
\end{equation}

The next theorem makes the transfer quantitative.  Its conditions are
one-dimensional numerical inequalities independent of the length $M$.

\begin{theorem}[Local-entropy transfer]
\label{thm:entropy-transfer}
Let $W$ be an $(m,h,b)$-block-local square-free source over an alphabet of
size $q$.  Fix $0<\delta\le1/2$ and let $L=\lfloor1/\delta\rfloor$.  Suppose
that there is an $a>0$ such that
\begin{equation}\label{eq:transfer-conditions}
  L\ge2m+1,
  \qquad e^{-aL}\le\frac12,
  \qquad
  \Phi_q(\delta)+\frac bL+a+2D_L(a)\le h.
\end{equation}
Then the support of $W$ contains a word $S$ such that
\begin{equation}\label{eq:transfer-equal-gap}
  \LCS(S[t,t+\ell),S[t+\ell,t+2\ell))
  \le(1-\delta)\ell
\end{equation}
for every available $t$ and $\ell$.  Consequently, $S$, and every factor of
$S$, is an $\eps$-synchronization string for every
$1-\delta/(2+\delta)<\eps<1$.
\end{theorem}

The condition in \cref{thm:entropy-transfer} has a direct interpretation.
The source supplies conditional entropy at rate $h$, whereas
$\Phi_q(\delta)$ bounds the exponential growth rate of the relevant deletion
balls.  The remaining terms are finite-scale losses: $b/L$ accounts for the
boundary blocks, and $a+2D_L(a)$ pays for dependencies among overlapping
intervals.  The proof, together with an explicit evaluation of $D_L(a)$,
appears in Appendix~\ref{app:general-theorem}.

\begin{definition}[Uniform occurrence-wise square-free substitution]
\label{def:branching-system}
Let $A,Q$ be finite alphabets, let $q=|Q|$, and let $m,K$ be positive
integers.  An $(m,K)$-uniform occurrence-wise substitution consists of a
language $\cT\subseteq A^*$ containing a word of every positive length and,
for each $a\in A$, a family
$\cC_a=\{C_{a,1},\ldots,C_{a,K}\}\subseteq Q^m$ of $K$ distinct codewords.
It is \emph{square-free} if, for every
$T=T_1\cdots T_N\in\cT$ and every occurrence-wise choice vector
$R\in[K]^N$, the concatenation
$C_{T_1,R_1}\cdots C_{T_N,R_N}$ is square-free.
\end{definition}

\begin{corollary}[A general substitution criterion]
\label{thm:branching-system}
Every substitution in \cref{def:branching-system} with $K\ge2$ generates
$Q$-ary synchronization strings of every length for one fixed parameter
$\eps<1$.
More precisely, for every sufficiently small $\delta>0$ and every requested
length, there is an output factor satisfying \cref{eq:transfer-equal-gap};
hence every $\eps$ with $1-\delta/(2+\delta)<\eps<1$ is valid.
\end{corollary}

\begin{proof}
For a requested length $n$, fix an outer word of length
$N=\lceil n/m\rceil$ and choose its occurrence variables independently and
uniformly from $[K]$.  The resulting $mN$-character word is a block-local
source with
\begin{equation}\label{eq:substitution-source-parameters}
  h=\frac{\log K}{m},
  \qquad b=2\log K.
\end{equation}
Indeed, an interval of length $\ell$ contains at least $\ell/m-2$ complete
macroblocks, and their distinct images give conditional point mass at most
$K^{-(\ell/m-2)}=e^{-h\ell+b}$.  Every outcome is square-free by definition.

Choose $a=h/4$.  As $\delta\downarrow0$, one has $L\to\infty$,
$\Phi_q(\delta)\to0$, and $b/L\to0$.  Moreover, $e^{-aL}\to0$ and the
geometric tail $D_L(a)$ in
\cref{eq:transfer-dependency-tail} tends to zero.  Thus all conditions in
\cref{eq:transfer-conditions} hold for every sufficiently small $\delta$.
Apply \cref{thm:entropy-transfer} and take a prefix of the required length.
\end{proof}

The six-codeword proof is a direct instance of the transfer theorem.

\begin{proposition}[The six-codeword instance]
\label{prop:transfer-small-family}
The local-entropy transfer theorem applies to the 18-uniform Brinkhuis source
with $\delta=1/1000$ and $a=1/80$.  Consequently, this source contains words
with normalized edit-distance gap at least $1/2001$.  In particular, it
yields the $2001/2002$-synchronization strings used in
\cref{thm:small-family}.
\end{proposition}

\begin{proof}
Here $q=3$, $m=18$, $h=(\log2)/18$, $b=2\log2$, and $L=1000$.
By
\cref{lem:volume-arithmetic,lem:elementary-logs,lem:neighbor-arithmetic}
we have $\Phi_3(\delta)<9/500$, $b/L<1/500$, $D_L(a)<1/490$, and
$h>23/600$.  For the bound on $D_L(a)$, take $\ell=L$ in
\cref{lem:neighbor-arithmetic} and discard its nonnegative $r=L$ term.
Moreover, $L\ge2m+1$ and $e^{-aL}=e^{-25/2}<1/2$.  Finally,
\[
  \Phi_3(\delta)+\frac bL+a+2D_L(a)
  <\frac9{500}+\frac1{500}+\frac1{80}+\frac1{245}<h,
\]
where the sum is smaller than $23/600$ by \cref{lem:final-margins}.  Thus
\cref{eq:transfer-conditions} holds, and \cref{thm:entropy-transfer} gives
normalized edit-distance gap $1/2001$.  Taking the smaller strict gap
$1/2002$ gives synchronization parameter $2001/2002$.
\end{proof}

\subsection{An extremal square-free consequence}

The transfer theorem also applies to a substitution developed for a
different extremal problem.  A square-free word is \emph{extremal} if
inserting any letter at any position creates a square.  Grytczuk,
Kordulewski, and Niewiadomski introduced these words, and Mol and Rampersad
constructed them through occurrence-wise choices on a twelve-vertex Thue
digraph
\cite{GrytczukKordulewskiNiewiadomski2020,MolRampersad2021}.

\begin{theorem}[Extremal square-free synchronization strings]
\label{thm:extremal-sync}
For every integer $n\ge2493$, there is a ternary word of length $n$ that is
simultaneously extremal square-free and a
$200001/200002$-synchronization string.
\end{theorem}

The Mol--Rampersad substitution offers images $Q_x$ and $R_x$ of lengths
$41$ and $52$ at each occurrence $x$.  For two consecutive occurrences
$x,y$, restrict the choices to $Q_xR_y$ and $R_xQ_y$.  These two words are
distinct and both have length $93$, so every pair exposes
one independent bit with fixed character boundaries.  A balanced
representation $n-98=41a+52b$ with $|a-b|\le46$ leaves only a constant number
of unpaired images.  Consequently, an interval of length $\ell$ contains
$\ell/93-O(1)$ complete random blocks, while every outcome remains extremal
square-free.  Applying \cref{thm:entropy-transfer} with
$\delta=1/100000$ gives the stated parameter.  The complete length-balancing,
conditional-entropy, and numerical arguments appear in
Appendix~\ref{app:extremal-square-free}.

\subsection{A cyclic consequence}

The same mechanism also applies when no distinguished starting position is
available.  Following \cite[Definition~2.6]{ChengEtAl2019}, a word is an
\emph{$\eps$-synchronization circle} if each of its cyclic rotations is an
$\eps$-synchronization string.  We call a word \emph{circularly square-free}
if each cyclic rotation is square-free.

\begin{lemma}[Circular preservation]
\label{lem:circular-preservation}
Let a substitution satisfy \cref{def:branching-system}, and suppose that
$\cT$ contains every square-free word over $A$.  If $T\in\cT$ is
circularly square-free and has length at least two, then every
occurrence-wise substitution output of $T$ is circularly square-free.
\end{lemma}

\begin{proof}
Fix the codeword chosen at every occurrence and write the output as
$W=C_0\cdots C_{N-1}$.  Consider a character rotation of $W$ whose cut lies
inside $C_i$.  This rotation is a factor of the substitution output of
$U=T_iT_{i+1}\cdots T_{i-1}T_i$, where the indices are cyclic and the
repeated final occurrence uses the same
codeword $C_i$.

It remains to check that $U$ is square-free.  Every proper factor of $U$ has
length at most $N$ and is a factor of a rotation of $T$.  If the whole word
$U$ were a square $XX$ of length $2p=N+1$, then its first and last letters,
both equal to $T_i$, would force positions $p$ and $p+1$ of $U$ to be equal:
position $p$ agrees with the last position and position $p+1$ with the first.
These adjacent positions lie in the displayed rotation of $T$, contradicting
circular square-freeness.  Thus $U$ is square-free.  The substitution
guarantee makes its output square-free, and so the chosen rotation of $W$,
being a factor, is square-free as well.
\end{proof}

\begin{theorem}[Cyclic substitution criterion]
\label{thm:cyclic-substitution}
Suppose a substitution in \cref{def:branching-system} has $K\ge2$ and its
outer language $\cT$ contains every square-free word over $A$.  There is a
fixed $\eps<1$ such that every circularly square-free outer word of length at
least two has an occurrence-wise output that is an
$\eps$-synchronization circle.
\end{theorem}

\begin{proof}
Put $h=(\log K)/m$ and $b=2\log K$, and choose fixed $\delta,a>0$ for which
the conditions of \cref{thm:entropy-transfer} hold.  Such a choice exists by
the limiting argument in the proof of \cref{thm:branching-system}.  Let
$L=\lfloor1/\delta\rfloor$, put $M=mN$, and view the independent macroblock
choices on a circle of circumference $M$.  For every cyclic start $t$ and
$L<\ell\le M/2$, let $E_{t,\ell}$ be the event that the two consecutive
length-$\ell$ arcs at $t$ have LCS greater than $(1-\delta)\ell$.

A cyclic arc of length $\ell$ contains at least $\ell/m-2$ complete
macroblocks.  Conditioning on every choice except those complete blocks in
the first arc fixes the second arc and gives the same conditional point-mass
bound as in \cref{lem:general-conditional-mass}.  Hence the deletion-ball
argument in \cref{lem:general-event-tail} gives the identical event
probability bound.

Declare an event to depend on every macroblock meeting its two-arc support.
For fixed lengths $\ell$ and $r$, the number of cyclic starts of a neighboring
length-$r$ event is at most
$\min\{M,2\ell+2r+2m+1\}\le 2\ell+2r+2m+1$.  Thus the cyclic dependency
neighborhood is no larger than the linear one in
\cref{lem:general-neighbor-count}.  The same charges and the same local-lemma
calculation avoid every event with $\ell>L$.

By \cref{lem:circular-preservation}, every output is circularly square-free.
Consequently, at $\ell\le L$ the two cyclic arcs are unequal and have LCS at
most $\ell-1\le(1-\delta)\ell$.  We have therefore enforced the equal-length
inequality at every cyclic start and every $\ell\le M/2$.  After cutting the
circle at any character position, these are exactly all equal-length pairs
needed in the resulting length-$M$ word.  The proof of
\cref{lem:equal-to-full} now gives non-strict normalized edit distance at
least $\delta/(2+\delta)$ in every rotation.  Any fixed strictly smaller gap
works, proving the claim.
\end{proof}

\begin{corollary}[Ternary synchronization circles]
\label{cor:ternary-circles}
For every integer $N\ge18$, there is a ternary
$2001/2002$-synchronization circle of length $18N$.  Hence three is the
minimum constant alphabet size supporting synchronization circles of
unbounded length for some fixed parameter below one.
\end{corollary}

\begin{proof}
Circularly square-free ternary words exist at every length $N\ge18$
\cite[Theorem~1]{Shur2010Circular}.  Apply \cref{thm:cyclic-substitution} to the
18-uniform family.  The explicit calculation in \cref{sec:ternary-proof}
uses $L=1000$; its event bounds and dependency estimates are unchanged on the
circle, as the preceding proof shows.  Thus every rotation has non-strict
normalized edit-distance gap $1/2001$, and the strictly smaller gap $1/2002$
gives parameter $2001/2002$.  Finally, every cut of a binary synchronization
circle is a binary synchronization string, so the binary obstruction of
\cite[Section~1.3.2]{ChengEtAl2019} rules out unbounded length over two symbols.
\end{proof}

%% file: sections/05_quantitative_summary.tex
\subsection{A stronger explicit parameter}
\label{sec:quantitative-summary}

Returning to the linear setting, the 54-uniform Brinkhuis family gives a
stronger numerical parameter.  Specifically, \cref{thm:quantitative-intro} gives
ternary $\eps$-synchronization strings of every finite length and one infinite
such string for every $\eps>215/216$; in particular, one may take
$\eps=216/217$.

The mechanism is the same conditional-entropy argument, with three
quantitative refinements.  First, the 54-uniform family supplies
952 choices per outer letter instead of two.  Second, near the analytic
cutoff we retain the exact number of visible choices and exclude the few
remaining one- and two-deletion obstructions by finite codebook checks.
We also count only fillings with unequal neighboring letters and restrict
matching runs according to their block phases.  At equal phases, a long
run would force a square in the outer word; at other phases, short
codeword-factor comparisons bound the run length.
Third, we place unequal adjacent intervals directly in the local lemma,
avoiding the loss in the generic equal-to-unequal conversion.  These steps
give the non-strict normalized edit-distance lower bound $1/216$, and
therefore every strict gap below $1/216$.
Appendix~\ref{sec:quantitative-strengthening} gives the probabilistic proof.
The bounded-shell statements and numerical inequalities are proved in
Appendix~\ref{app:numerics} by finite codebook calculations and analytic tail
bounds.

%% file: sections/06_abundance_algorithm.tex
\section{Abundance and randomized construction}\label{sec:extensions}

The same local-lemma charges bound both the number of valid words and the
expected work needed to find one.  The product conclusion in \cref{thm:lll}
lower-bounds the probability of avoiding all bad events; injectivity of the
substitution converts this probability into a count.  Moser--Tardos
resampling then gives a Las Vegas construction.

\subsection{Exponential abundance}

\begin{theorem}[Exponential abundance]\label{thm:general-abundance}
Under the hypotheses of \cref{thm:branching-system}, the cutoff $L$ may be
chosen so that, for all sufficiently large $n$, at least
$\exp((\log K)n/(3m))$ words in $Q^n$ satisfy the resulting fixed
synchronization inequality.
\end{theorem}

\begin{proof}
Put $h=(\log K)/m$ and $a=h/4$.  Choose $\delta>0$ small enough that
\cref{eq:transfer-conditions} holds and, with
$L=\lfloor1/\delta\rfloor$, the geometric tail
$Q_L:=\sum_{\ell>L}e^{-a\ell}$ is at most $h/4$.  Such a choice exists
because $L\to\infty$ and $Q_L\to0$ as $\delta\downarrow0$, whereas $a$
remains fixed.

Fix an outer word of length $N$ and let $M=mN$ be the full output length.
For each $\ell$ there are fewer than $M$ bad events, one for each possible
start.  Thus the charges $x(E_{t,\ell})=e^{-a\ell}$ used in the proof of
\cref{thm:entropy-transfer} have total charge at most $MQ_L$.  Every charge
is at most $1/2$, so the product form of the local lemma, together with
\cref{lem:small-charge-product}, shows that all bad events are avoided with
probability at least $e^{-2MQ_L}$, and hence at least $e^{-hM/2}$.
Since the product space has $K^N=e^{hM}$ equally likely points, at least
$e^{hM/2}$ choice vectors avoid every bad event.  They give distinct full
outputs: at each macroblock boundary, the codeword identifies the choice
made for the fixed outer letter.

For an arbitrary requested length $n$, take $N=\lceil n/m\rceil$ and
$M=mN$.  The length-$n$ prefix contains the first $N-1$ macroblocks in full
and hence determines all but possibly the last choice.  Every prefix thus
has at most $K=e^{hm}$ full-output preimages.  The number of distinct good
prefixes is therefore at least $e^{h(M/2-m)}$.  If $n\ge6m$, then
$M/2-m\ge n/2-n/6=n/3$, giving the claimed count.  Each prefix inherits the equal-length
inequalities, and \cref{lem:equal-to-full} gives the same fixed
synchronization parameter.
\end{proof}

For the 54-uniform family, the same charges that prove the improved
synchronization parameter retain almost all of the occurrence-wise choice
entropy.

\begin{corollary}[Computer-assisted explicit ternary count]\label{cor:explicit-count}
For every $n\ge504$, at least $e^{n/9}$ ternary words of length $n$ satisfy
\cref{eq:quantitative-gap}.
\end{corollary}

\begin{proof}
Let $M=54N$ be a full output of the family in
\cref{prop:large-codebook}.  At every ordered length type $(a,b)$ there are
fewer than $M$ event starts.  The product conclusion of \cref{thm:lll}, with
the charges from \cref{eq:unequal-charges}, shows that the probability of
avoiding every bad event is greater than $e^{-M/1000}$.  Here we use
$-\log(1-x)\le x/(1-x)$ and \cref{lem:unequal-charge-envelopes}.
The $952^N$ choice vectors are equally likely and give distinct outputs, so
more than $952^N e^{-M/1000}$ full outputs are good.  This exceeds $e^{M/8}$
because $(\log952)/54-1/1000>1/8$ by
\cref{lem:large-family-log-bounds}.

For a requested length $n$, take $M=54\lceil n/54\rceil$.  A length-$n$
prefix determines all choices except possibly the last, so at most 952 full
outputs have the same prefix.  The number of distinct good prefixes is
therefore larger than $e^{M/8}/952>e^{n/8-7}\ge e^{n/9}$.  Here
$\log952<7$, $M\ge n$, and the last inequality is equivalent to $n\ge504$.
Prefixes inherit every adjacent-interval inequality, so each counted word
satisfies \cref{eq:quantitative-gap}.
\end{proof}

\subsection{A Las Vegas construction}

\begin{theorem}[Randomized polynomial-time construction]\label{thm:las-vegas}
Suppose an $(m,K)$-uniform occurrence-wise square-free substitution with $K\ge2$ is fixed
and an outer word in $\cT$ of length $N$ can be generated in
$\operatorname{poly}(N)$ time.  Then the strings in
\cref{thm:branching-system} can be generated by a Las Vegas randomized
algorithm in expected polynomial time.
\end{theorem}

\begin{proof}
For the general statement, generate the outer word, sample all occurrence
choices independently, and scan the equal-block events used in the proof of
\cref{thm:branching-system}.  Each event can be recognized exactly in
polynomial time by the standard dynamic program for LCS.  Whenever an event
is violated, resample exactly the macroblock choices meeting its support and
restart the scan.  The event scopes and charges satisfy the local-lemma
criterion by \cref{lem:general-lll-closure}, so the Moser--Tardos theorem
applies.

Write $M$ for the full output length and $Z$ for the total number of
resamplings.  Moser--Tardos bounds $\mathbb E Z$ by
$\sum_E x(E)/(1-x(E))$.  Every charge is at most $1/2$, and the geometric
charge tail is fixed, so this sum is at most $2\sum_E x(E)=O(M)$.
In particular, the algorithm terminates almost surely.  At termination no
long equal-block event remains; square-freeness handles the omitted short
scales, and \cref{lem:equal-to-full} gives the claimed synchronization
parameter.  This proves the general statement.

\end{proof}

For the 54-uniform source, a faster implementation is possible after adding
constraints on distant intervals.  These constraints are useful even if
only adjacent synchronization is wanted: a long bad alignment can be
shortened, but its two surviving intervals need not be adjacent.  Including
distant pairs makes this reduction closed under taking subalignments.

\Needspace{8\baselineskip}
\begin{theorem}[Long-distance synchronization and construction]
\label{thm:long-distance-construction}
For every $n\ge1$, a ternary word satisfying \cref{eq:quantitative-gap}
can be generated in expected $O(n^2\log^3(n+2))$ time using $O(n^2)$
machine words of space.  The algorithm is Las Vegas and uses the
computer-assisted bounds of Appendix~\ref{sec:quantitative-strengthening}.
Its output also satisfies
\[
 \ED(S[i,j),S[k,\ell))\ge\frac{(j-i)+(\ell-k)}{216}
\]
whenever $1\le i<j\le k<\ell\le n+1$ and
$(j-i)+(\ell-k)\ge30000\log n$.
\end{theorem}

\begin{proof}
Put $M=54\lceil n/54\rceil$, generate a square-free outer word of length
$M/54$, and choose its occurrence-wise substitutions independently.
\Cref{lem:combined-lll,lem:alignment-truncation} give a variable local lemma
whose events involve two intervals of total length at most
$B=2(10000+1200\lceil\log_2 M\rceil)+4$.  It includes adjacent violations
and distant violations above the specified cutoff.  Avoiding these events
implies all the stated inequalities, including the unbounded-length
adjacent ones.  The expected number of resamplings is $O(M)$.
The probability and cutoff proofs are in Appendix~\ref{app:long-distance}.

It remains to implement event selection without rescanning all intervals.
Set $H=\min\{B,M\}$.  For each ordered pair of starting positions $p<q$,
compute the ordinary prefix--prefix LCS dynamic program on the windows
starting there, each of length at most $H$.  Truncate the first window at
$q$ so that the two intervals are disjoint.  The table contains the LCS
value for every pair of prefixes of these windows.  Inspect the eligible
lengths $a,b$ with $a+b\le B$, distinguishing adjacent pairs from distant
pairs at or above the cutoff.  All tests are exact integer comparisons:
an event is bad precisely when
$216(a+b-2\LCS)<a+b$.  If a bad event exists, store one such pair of
lengths; otherwise store a null entry.  The cost per pair of starts is
$O(H^2)$ time and $O(H)$ temporary space using two rows of the dynamic
program.  Initialization costs $O(M^2H^2)$.

Maintain a linked list of the nonnull entries, with a pointer from each
entry to its list position.  Select any entry on this list, retrieve its
stored event, and resample exactly the macroblock variables meeting that
event's two intervals.  This is an original bad event of the local lemma,
not the union of all events represented by the entry.  Its resampled blocks
form at most two contiguous regions of total length $O(H)$, since the
block length is fixed.

Only pairs whose left or right window meets one of these regions need to be
recomputed.  For a changed region $[u,v)$, a window of length at most $H$
can meet it only if its start lies in $[u-H+1,v-1]$.  There are $O(H)$ such
starts, and at most $M$ choices for the other start.  Hence $O(MH)$ cached
entries are affected.  Enumerate these entries using a Boolean mark per
cache entry to suppress duplicates, and record the marked entries in a
temporary list.  Recompute each listed entry and update its linked-list
position in constant time once its new value is known.  Finally, clear
the marks by traversing the temporary list.  Marking and clearing take
$O(MH)$ time, so the cost per resampling is $O(MH^3)$.
All other entries remain valid because none of their inspected symbols
changed.  Thus the list is empty exactly when no retained bad event exists.

The expected total time is $O(M^2H^2)+O(M)O(MH^3)
=O(n^2\log^3(n+2))$.  The cache, pointers, marks, and temporary list use $O(M^2)$
words; a word has $O(\log(n+2))$ bits.  The temporary LCS rows and the
output need less space.  Sampling from the fixed 952-element choice set
has constant expected cost per variable, and the Thue outer word can be
generated in linear time~\cite{Thue1906}.  Moser--Tardos ensures almost-sure
termination under this event-selection rule.  At termination, return the
length-$n$ prefix, which inherits the required inequalities.
\end{proof}

%% file: sections/07_discussion.tex
\section{Consequences and open directions}\label{sec:discussion}

This work determines the exact alphabet threshold for synchronization
strings.  Ternary synchronization strings exist at every length for one
fixed parameter below one, while the known binary obstruction shows that
three symbols are necessary.  The construction stays within the ternary
square-free language: occurrence-wise choices in a Brinkhuis family supply
the long-range information that the previous four-symbol construction stored
in a separate marking layer.

The local-entropy transfer theorem separates this mechanism from the
particular substitution.  It compares two rates: the conditional min-entropy
that remains in an interval after the rest of the source is exposed, and the
entropy of a high-LCS deletion ball.  Synchronization follows when the former
dominates the latter, subject to bounded-scale and dependency conditions.
Brinkhuis families meet this criterion, and the unequal-length Mol--Rampersad
substitution gives a second application: extremal square-freeness and a fixed
synchronization gap can hold simultaneously.

The optimal ternary synchronization gap remains open.  The larger Brinkhuis
family improves the explicit parameter by supplying more choices per output
character, but its short scales require finite deletion-shell checks.  A
source whose blocks already exclude near-repetitions could improve the gap
and replace these checks by a structural argument.  The transfer theorem
identifies the long-scale requirement for such a source: its interval
conditional entropy must dominate the corresponding deletion-ball rate.

The long-distance extension also improves construction time.  A matching-path
reduction restricts attention to logarithmic total lengths, and caching the
events by their two starts gives expected $O(n^2\log^3(n+2))$ time.
The quadratic collection of start pairs remains a bottleneck.  Near-linear-time
and locally accessible ternary constructions remain further goals.

%% file: appendix/C_general_theorem.tex
\section{Proof of the local-entropy transfer theorem}
\label{app:general-theorem}

This appendix proves \cref{thm:entropy-transfer}.  Fix an
$(m,h,b)$-block-local square-free source $W\in Q^M$, where $|Q|=q$, together
with $0<\delta\le1/2$ and the charge exponent $a$ from the theorem.
Throughout, put $L=\lfloor1/\delta\rfloor$,
$\gamma=h-\Phi_q(\delta)-b/L$, and $D=D_L(a)$.
The hypotheses imply $L\ge2m+1\ge3$ and $0<a<\gamma$.
Writing $z=e^{-a}$ and summing the two geometric series gives
\begin{equation}\label{eq:transfer-dependency-closed-form}
  D_L(a)
  =3z^{L+1}\left(
       \frac1{1-z}
       +\frac{L+1-Lz}{L(1-z)^2}
     \right).
\end{equation}

\subsection{The probability of a bad interval}

\begin{lemma}[Deletion-ball entropy rate]
\label{lem:general-volume-detail}
For every integer $\ell>L$ and every $y\in Q^\ell$,
\[
 \left|\left\{x\in Q^\ell:
 \LCS(x,y)>(1-\delta)\ell\right\}\right|
 \le \exp\!\bigl(\Phi_q(\delta)\ell\bigr).
\]
\end{lemma}

\begin{proof}
By \cref{lem:lcs-ball}, the cardinality in question is at most
\begin{equation}\label{eq:transfer-deletion-sum}
  \sum_{0\le d<\delta\ell}\binom{\ell}{d}^{\!2}q^d.
\end{equation}
For $0\le d<\delta\ell$, the standard entropy estimate gives
\[
  \binom{\ell}{d}
  \le \exp\!\left(\ell\mathsf H\!\left(\frac d\ell\right)\right)
  \le \exp\!\bigl(\ell\mathsf H(\delta)\bigr).
\]
The first inequality follows by taking
$p=d/\ell$ in
$1=(p+(1-p))^\ell\ge\binom{\ell}{d}p^d(1-p)^{\ell-d}$;
the cases $d=0,\ell$ follow by continuity.
The second inequality uses $d/\ell<\delta\le1/2$ and the monotonicity of
binary entropy on $[0,1/2]$.  Also $q^d\le e^{\delta\ell\log q}$.

Because $\ell>L=\lfloor1/\delta\rfloor$, one has
$z:=\delta\ell>1$.  The number of integers in the sum is $\lceil z\rceil$,
and $\lceil z\rceil\le z+1\le2^z=e^{\delta\ell\log2}$.  Multiplying this
count by the bound for one summand shows that
\cref{eq:transfer-deletion-sum} is at most
$\exp(\ell(2\mathsf H(\delta)+\delta\log q+\delta\log2))
=\exp(\Phi_q(\delta)\ell)$, as claimed.
\end{proof}

Combining this count with the source's conditional entropy bounds each
bad-event probability.  For every available start $t$ and every integer $\ell>L$, let
$E_{t,\ell}$ be the event
\begin{equation}\label{eq:transfer-bad-event}
  \LCS(W[t,t+\ell),W[t+\ell,t+2\ell))
  >(1-\delta)\ell.
\end{equation}

\begin{lemma}[Conditional point mass]
\label{lem:general-conditional-mass}
Let $X=W[t,t+\ell)$ and $Y=W[t+\ell,t+2\ell)$.  Condition on every block
variable except those whose blocks are contained completely in the interval
of $X$.  Then $Y$ is fixed and, under every conditioning of positive
probability,
\[
  \Pr[X=x\mid\text{the conditioning}]
  \le e^{-h\ell+b}
  \qquad\text{for every }x\in Q^\ell.
\]
\end{lemma}

\begin{proof}
Every unconditioned variable affects a block contained in the interval of
$X$, so it affects no character of the disjoint interval of $Y$.  All block
variables affecting $Y$, including the variable of a block crossing the
common boundary, are therefore conditioned, and $Y$ is fixed.  The stated
point-mass bound is exactly \cref{eq:block-local-min-entropy} applied to the
interval of $X$.
\end{proof}

\begin{lemma}[Bad-event tail]
\label{lem:general-event-tail}
Every event in \cref{eq:transfer-bad-event} satisfies
$\Pr[E_{t,\ell}]\le e^{-\gamma\ell}$.
\end{lemma}

\begin{proof}
Condition as in \cref{lem:general-conditional-mass}.  The right interval has
a fixed value $y$.  By \cref{lem:general-volume-detail}, at most
$e^{\Phi_q(\delta)\ell}$ values of $X$ make the event true, and each has
conditional probability at most $e^{-h\ell+b}$.  Hence
\[
  \Pr[E_{t,\ell}\mid\text{the conditioning}]
  \le \exp\!\bigl(-(h-\Phi_q(\delta))\ell+b\bigr).
\]
Since $\ell>L$ and $b\ge0$, one has $b\le(b/L)\ell$.  The last display is
therefore at most $e^{-\gamma\ell}$.  This bound is independent of the
conditioned values, so averaging proves the same unconditional estimate.
\end{proof}

\subsection{Dependencies across all scales}

Declare that $E_{t,\ell}$ depends on the variables of all blocks intersecting
its character support $[t,t+2\ell)$.  Join two events when their declared
variable sets intersect.

\begin{lemma}[Neighbors at one scale]
\label{lem:general-neighbor-count}
Fix an event of length $\ell$.  For every $r>L$, the number of its neighbors
of length $r$ is at most $2\ell+2r+2m+1$.
If $L\ge2m+1$, this number is at most $3(\ell+r)$.
\end{lemma}

\begin{proof}
Let the fixed event start at $t$ and another length-$r$ event start at $u$.
If their variable sets intersect, one block of length at most $m$ meets both
character supports.  Consequently, $t-2r-m<u<t+2\ell+m$.  This open
interval contains at most $2\ell+2r+2m+1$ integer starts, even
before starts outside the word are discarded.  If $L\ge2m+1$ and
$\ell,r>L$, then $2m+1\le\ell+r$, giving the second bound.
\end{proof}

Assign the charge
\begin{equation}\label{eq:transfer-charge}
  x(E_{t,\ell})=e^{-a\ell}.
\end{equation}

\begin{lemma}[Total neighboring charge]
\label{lem:general-neighbor-charge}
For every event $E_{t,\ell}$,
\[
  \sum_{F\sim E_{t,\ell}}x(F)\le D\ell,
\]
where $D=D_L(a)$ is defined in \cref{eq:transfer-dependency-tail}.
\end{lemma}

\begin{proof}
At scale $r$, \cref{lem:general-neighbor-count} gives at most
$3(\ell+r)$ neighbors.  Extending the finite sum to every integer $r>L$
only increases it.  Since $\ell>L$, we have
$\ell+r\le\ell(1+r/L)$.  Thus the total neighboring charge is at most
$3\sum_{r>L}(\ell+r)e^{-ar}\le D\ell$ by the definition of $D$.
\end{proof}

\begin{lemma}[Local-lemma closure]
\label{lem:general-lll-closure}
Under \cref{eq:transfer-conditions}, every event in
\cref{eq:transfer-bad-event} satisfies the asymmetric local-lemma criterion.
\end{lemma}

\begin{proof}
Every event has length greater than $L$, so
$x(E_{t,\ell})=e^{-a\ell}<e^{-aL}\le1/2$.  The same upper bound holds for
every neighboring charge.  Using $\log(1-u)\ge-2u$ for $0\le u\le1/2$,
the logarithm of the neighboring product is at least
$-2\sum_{F\sim E_{t,\ell}}x(F)\ge-2D\ell$, by
\cref{lem:general-neighbor-charge}.  Since $a+2D\le\gamma$, it follows that
\[
  x(E_{t,\ell})
  \prod_{F\sim E_{t,\ell}}(1-x(F))
  \ge e^{-(a+2D)\ell}
  \ge e^{-\gamma\ell}
  \ge\Pr[E_{t,\ell}],
\]
where the last inequality is \cref{lem:general-event-tail}.  This is the
criterion in \cref{eq:lll-criterion}.
\end{proof}

We can now combine the long-interval estimate with square-freeness at the
remaining scales.

\begin{proof}[Proof of \cref{thm:entropy-transfer}]
The family \cref{eq:transfer-bad-event} is finite.  Each event is determined
by its declared block variables, and events with disjoint declared variable
sets are functions of disjoint families of independent variables.  Thus the
graph above is a valid variable-dependency graph.  By
\cref{lem:general-lll-closure,thm:lll}, some outcome $S$ avoids every event
with $\ell>L$.  Hence \cref{eq:transfer-equal-gap} holds at all these scales.

It remains to consider $1\le\ell\le L$.  Every outcome of the source is
square-free.  The two adjacent length-$\ell$ intervals cannot be equal, for
otherwise their concatenation would be a square.  Two different words of
the same length have LCS at most $\ell-1$.  Moreover,
$\delta\ell\le\delta L\le\delta\lfloor1/\delta\rfloor\le1$.  Therefore
\[
  \LCS(S[t,t+\ell),S[t+\ell,t+2\ell))
  \le\ell-1\le(1-\delta)\ell,
\]
which proves \cref{eq:transfer-equal-gap} at every scale.

Apply \cref{lem:equal-to-full} with $\eta=\delta$.  Every pair of adjacent
intervals has normalized insertion--deletion distance at least
$\delta/(2+\delta)$.  Thus every strictly smaller gap, equivalently every
parameter $\eps$ with $1-\delta/(2+\delta)<\eps<1$ satisfies the strict
definition of a synchronization string.  Finally, factors inherit all
equal-block inequalities and hence the same conclusion.
\end{proof}

%% file: appendix/E_extremal_square_free.tex
\section{Extremal square-free synchronization strings}
\label{app:extremal-square-free}

An extremal square-free word is a square-free word for which inserting any
letter at any position creates a square.  This appendix proves
\cref{thm:extremal-sync} by applying the local-entropy transfer theorem to a
multivalued substitution of Mol and Rampersad.

Let $\mathcal G$ be the twelve-vertex Thue digraph used in the construction
of extremal ternary square-free words.  Grytczuk, Kordulewski, and
Niewiadomski proved that $\mathcal G$ has arbitrarily long square-free walks
\cite[Lemma~8]{GrytczukKordulewskiNiewiadomski2020}.  For each vertex $x$,
Mol and Rampersad define words $Q_x$ and $R_x$ of lengths $41$ and $52$, as
well as a prescribed length-$49$ left cap at the first vertex and a
prescribed length-$49$ right cap at the last vertex.  Their construction has
the following property: if $w=xw'y$ is a square-free walk in $\mathcal G$ of
length at least two, then every word obtained by choosing $Q_z$ or $R_z$
arbitrarily and occurrence-wise at each occurrence $z$ of $w$, and then
adding the two prescribed caps, is extremal square-free
\cite[Corollary~3.3]{MolRampersad2021}.

The twelve pairs $(Q_x,R_x)$ are obtained by applying the six permutations
of the alphabet, with and without reversal, to the base pair
\begin{align*}
 Q&=\mathtt{abacbabcacbacabacbcacbacabcbabcabacbcabcb},\\
 R&=\mathtt{abacabcacbacabcbabcacbacabacbcacbacabcbabcabacbcabcb}.
\end{align*}

\begin{lemma}[Equal-length pairing]
\label{lem:extremal-equal-pairing}
For every ordered pair of vertices $x,y$ of $\mathcal G$, the words
$Q_xR_y$ and $R_xQ_y$ are distinct and have length $93$.
\end{lemma}

\begin{proof}
Both lengths equal $41+52=93$.  If the two concatenations were equal, their
first $41$ letters would give $Q_x=R_x[1,42)$.  For a vertex obtained
without reversal, this is impossible because the base
words $Q$ and $R$ first differ in their fifth letter.  For a reversed vertex,
it would say that $Q$ is the length-$41$ suffix of $R$; direct comparison of
the two displayed words shows that their first letters already differ.
Alphabet permutations preserve both inequalities.
\end{proof}

\begin{lemma}[Balanced lengths]
\label{lem:balanced-41-52}
For every integer $M\ge2395$, there are nonnegative integers $a,b$ such that
$M=41a+52b$ and $|a-b|\le46$.
\end{lemma}

\begin{proof}
Choose an integer $N\equiv 7M\pmod {11}$ minimizing $|N-2M/93|$.
Consecutive integers in this residue class
differ by $11$, and hence
\begin{equation}
  \left|N-\frac{2M}{93}\right|\le\frac{11}{2}.
  \label{eq:nearest-residue-extremal}
\end{equation}
Set $d=(2M-93N)/11$, $a=(N-d)/2$, and $b=(N+d)/2$.
Since $93\equiv5\pmod {11}$ and $5^{-1}\equiv9\pmod {11}$, the congruence
$N\equiv7M\pmod {11}$ implies $93N\equiv2M\pmod {11}$.  Thus $d$ is an
integer.  Reducing $11d=2M-93N$ modulo two shows that $d$ and $N$ have the
same parity, so $a$ and $b$ are integers.  Moreover,
$41a+52b=(93N+11d)/2=M$.  By \cref{eq:nearest-residue-extremal},
$|d|=|2M-93N|/11\le93/2$, and integrality improves this to $|d|\le46$.
Finally, if $M\ge2395$, then $N\ge2M/93-11/2>46\ge|d|$.  Therefore
$a=(N-d)/2$ and $b=(N+d)/2$ are nonnegative, and
$|a-b|=|d|\le46$.
\end{proof}

\begin{lemma}[A fixed-boundary extremal source]
\label{lem:extremal-source}
For every $n\ge2493$, there is a $(m,h,b_0)$-block-local square-free source
of length $n$, all of whose outcomes are extremal, with parameters
\[
  m=93,
  \qquad
  h=\frac{\log2}{93},
  \qquad
  b_0=\left(\frac{2490}{93}+2\right)\log2.
\]
\end{lemma}

\begin{proof}
Put $M=n-98$ and take $a,b$ from \cref{lem:balanced-41-52}.  Set
$N=a+b$.  Choose a square-free walk $w=z_1\cdots z_N$ in $\mathcal G$.
Such a walk exists because a prefix of a sufficiently long square-free walk
has every prescribed shorter length.

Let $p=\min\{a,b\}$.  Reserve the first $|a-b|$ occurrences of the walk for
deterministic images: use $Q$-images if $a>b$ and $R$-images if $b>a$.
Partition the remaining $2p$ occurrences into consecutive ordered pairs
$(x,y)$.  Independently at each pair, choose with equal probability one of
the words $Q_xR_y$ and $R_xQ_y$.
By \cref{lem:extremal-equal-pairing}, these are distinct words of length
$93$.  Finally, add the prescribed length-$49$ caps at the two ends.

Every outcome uses exactly $a$ copies of a $Q$-image and $b$ copies of an
$R$-image.  Its length is therefore $98+41a+52b=n$.  Every outcome belongs
to the Mol--Rampersad substitution applied to the same
walk $w$, and hence is extremal square-free.

The two choices above form a random macroblock
of length $93$.  The cap blocks and the at most $46$ unpaired images are
deterministic and have total length at most $2\cdot49+46\cdot52=2490$.
Because the unpaired images were placed first, the random length-$93$
macroblocks form one contiguous region.  The block partition is fixed before
the independent fair bits are sampled, and every block has length at most
$93$.

Let $I$ be a character interval of length $\ell$.  Remove its intersection
with the deterministic regions and the at most two random macroblocks cut by
its endpoints.  The interval then contains at least
$k\ge \ell/93-2490/93-2$ complete random macroblocks.  Condition on all
other block variables.  A
specified value of the word on $I$ fixes the fair bit in each of these $k$
blocks, because the full block lies in $I$ and its two possible values are
distinct.  Hence
\[
  \max_u \Pr[W[I]=u\mid\text{the conditioning}]
  \le 2^{-k}
  \le \exp\!\left(
      -\frac{\log2}{93}\ell
      +\left(\frac{2490}{93}+2\right)\log2
    \right).
\]
This is the required block-local point-mass bound.  The deterministic blocks
may be assigned degenerate independent variables.
\end{proof}

\begin{proof}[Proof of \cref{thm:extremal-sync}]
Apply \cref{thm:entropy-transfer} to the source in
\cref{lem:extremal-source}.  Put $\delta=1/100000$, $L=100000$, and take
charge exponent $\alpha=1/1000$.  The elementary estimates in
\cref{lem:elementary-logs} give $2/3<\log2<1$, $\log100000<12$, and
$\log6<2$.  Hence $h>2/279$ and $b_0<29$, while the entropy estimate
$\Phi_3(\delta)\le2\delta\log(e/\delta)+\delta\log6$ gives
$\Phi_3(\delta)<28/100000$.  Consequently,
\[
  \gamma=h-\Phi_3(\delta)-\frac{b_0}{L}
  >\frac2{279}-\frac{57}{100000}
  >\frac6{1000}.
\]

It remains to bound the dependency tail $D$ in
\cref{eq:transfer-dependency-closed-form}.  Write $z=e^{-\alpha}$.  Since
$1-e^{-x}\ge x/2$ and $1-e^{-x}\le x$ for $0\le x\le1$, we have
$1/2000\le1-z\le1/1000$.  Also $z^{L+1}<e^{-100}<10^{-40}$; the latter
inequality follows from the
first seven terms of the exponential series, which give $e^5>100$.  Moreover,
$(L+1-Lz)/(L(1-z)^2)=1/(1-z)+1/(L(1-z)^2)<2040$.  Substitution in
\cref{eq:transfer-dependency-closed-form} gives
$D<3\cdot10^{-40}(2000+2040)<10^{-35}$.  Therefore
$\alpha+2D<2/1000<\gamma$.  The remaining hypotheses are
immediate: $L\ge2m+1$ and $e^{-\alpha L}=e^{-100}<1/2$.  The transfer theorem
produces an outcome whose normalized insertion--deletion gap is at least
$\delta/(2+\delta)=1/200001$.  Taking the strictly smaller gap $1/200002$
gives synchronization parameter
$200001/200002$.  By \cref{lem:extremal-source}, the selected outcome remains
extremal square-free.
\end{proof}

%% file: appendix/D_all_length_circles.tex
\section{All sufficiently large synchronization-circle lengths}
\label{app:all-length-circles}

To obtain circle lengths not divisible by $18$, we delete a short factor
from an 18-uniform output.  We must then exclude squares across the new join
and retain enough randomness for the local lemma.  A stronger outer word
bounds the size of any new square, reducing the first task to $306$ finite
template--context cases.  Increasing the cutoff handles the choices fixed
near the join.

\begin{theorem}[Computer-assisted all sufficiently large circle lengths]
\label{thm:all-length-circles}
For every integer $n\ge10008$, there is a ternary
$5201/5202$-synchronization circle of length $n$.
\end{theorem}

\subsection{Circular Dejean words and phase recognition}

The exponent of a nonempty word is its length divided by its least period.  A
circular word is $7/4^+$-free if none of its factors has exponent strictly
greater than $7/4$.  Currie, Mol, and Rampersad proved that such ternary
circular words exist at every length at least $23$
\cite[Theorems~3.2--3.3]{CurrieMolRampersad2019}.  We use the form of their
inductive step.
They define morphisms of lengths $19$ and $23$ with
\begin{align*}
d_{19}(0)&=0120212012102120210,\\
d_{23}(0)&=01202120102012102120210,
\end{align*}
and obtain the other images by adding the input letter modulo three.  Their
proof shows that every length $N>555$ has a circular $7/4^+$-free word of the
form
\begin{equation}
  d_{19}(u)d_{23}(v),
  \qquad |u|\ge6,\qquad 2\le |v|\le20,
  \label{eq:dejean-mixed-form}
\end{equation}
where $(uv)$ is itself circularly $7/4^+$-free.

We also use a small recognizability property of the Brinkhuis codebook.

\begin{lemma}[Phase recognition]\label{lem:phase-recognition}
Every length-$14$ factor of an occurrence-wise output of the six codewords in
\cref{eq:codebook} determines the phase of its starting position modulo
$18$.  The six codewords are pairwise distinct.
\end{lemma}

\begin{proof}
A length-$14$ factor meets at most two consecutive macroblocks.  It therefore
occurs in an output of three outer letters.  There are $12$ square-free
ternary words of length three and eight occurrence-wise choice vectors, for
$96$ outputs.  Exact factor enumeration gives $198$ distinct length-$14$
factors, none occurring at two phases.  The bound is sharp: at length $13$
the three factors
\[
0121021201210,\qquad
1202102012021,\qquad
2010210120102
\]
occur at both phases $8$ and $15$.  Direct comparison of the six displayed
codewords also proves their distinctness.
\end{proof}

\subsection{A bounded defect}

Delete one cyclic factor of length $g$ from an occurrence-wise output and
join its endpoints.  We call the new join the \emph{defect seam}.  The
inherited macroblock phases jump by $g$ modulo $18$ at this seam.

\begin{lemma}[Interior seam exclusion]\label{lem:interior-seam-exclusion}
Suppose $g\not\equiv0\pmod {18}$.  A square crossing the defect seam cannot
place the seam at distance at least $14$ from its start, midpoint, and end.
\end{lemma}

\begin{proof}
Unwrap the square as $[a,a+2p)$, with midpoint $a+p$, and let $s$ be the
seam in its interior.  Suppose first that $a<s<a+p$.  The four intervals
\[
[s-14,s),\quad [s-14+p,s+p),\quad
[s,s+14),\quad [s+p,s+p+14)
\]
lie inside the square and avoid the seam.  The first pair is equal, as is the
second pair.  By \cref{lem:phase-recognition}, their starts have equal
phases.  The first translation by $p$ crosses the deleted gap in undeleted
coordinates and gives $p+g\equiv0\pmod {18}$; the second gives
$p\equiv0\pmod {18}$.  This contradicts the assumption on $g$.  If the
seam lies in the second copy, compare the same windows with their translates
by $-p$.
\end{proof}

The edge cases become finite because the outer circle avoids exponents just
above $7/4$.

\begin{lemma}[Equality segment next to the seam]
\label{lem:seam-equality-segment}
Unwrap a square of root length $p$ as $[a,a+2p)$ in a word obtained by
deleting a gap of length $g$, and suppose that the defect seam lies in the
interior of the square and within $13$ symbols of its start, midpoint, or
end.  If $p\ge27$, then the equality of the two roots contains a pair of
equal, seam-free segments of a common length $h\ge p-13$.  In the undeleted
word, the starts of these segments differ by either $p$ or $p+g$.
\end{lemma}

\begin{proof}
If the seam lies in the first root, write its distance from the start as
$d$, where $0<d<p$.  For root coordinates $0\le i<d$, the first position is
before the seam and its mate at distance $p$ is after it; in undeleted
coordinates their displacement is $p+g$.  For $d\le i<p$, both positions
are after the seam and their displacement is $p$.  Thus the root equality
splits into seam-free segments of lengths $d$ and $p-d$.  In this case the
seam can be within $13$ of the start or midpoint only; accordingly, one of
these two lengths is at least $p-13$.

If the seam lies in the second root, let $d$ be its distance from the
midpoint, where $0\le d<p$.  For $0\le i<d$, both corresponding positions are before the seam
and have undeleted displacement $p$; for $d\le i<p$, the two positions lie
on opposite sides and have displacement $p+g$.  The segment lengths are
again $d$ and $p-d$.  Proximity to the midpoint or end makes one of them at
least $p-13$.  These cases exhaust the possible position of an interior
seam.  Since $p\ge27$, the selected segment has length at least $14$.
\end{proof}

\begin{lemma}[Edge-root bound]\label{lem:dejean-edge-bound}
Let the outer circle be ternary and $7/4^+$-free, and delete a gap of length
$1\le g\le35$.  If a square crossing the defect seam puts the seam within
$13$ symbols of its start, midpoint, or end, then its root length is at most
$288$.
\end{lemma}

\begin{proof}
Let $p$ be the root length.  If $p\le26$, the conclusion is immediate, so
assume $p\ge27$.  By
\cref{lem:seam-equality-segment}, the two copies contain an ordinary equality
segment of length $h\ge p-13\ge14$.  By \cref{lem:phase-recognition}, the
undeleted displacement between its
two starts is $D=18k$ for some positive integer $k$.  This displacement is
either $p$ or $p+g$, and hence $h\ge18k-48$.

The equality segment contains at least
$q\ge\lfloor(h-17)/18\rfloor\ge k-4$ complete macroblocks.  The starts of
the corresponding blocks are $k$ outer
positions apart.  Since the codewords are distinct, the two length-$q$
outer factors are equal.

Let $N$ be the outer circumference.  The original square satisfies
$2p\le18N-g$.  If $D=p$, then $2k<N$.  If $D=p+g$, then
$2k<N+2$.  Thus either $2k\le N$ or $2k=N+1$.
For $k\ge5$ we must have $q<k$: otherwise $2k\le N$ gives an outer
square, while $2k=N+1$ gives a length-$(2k-1)$ conjugate of period at most
$k$ and exponent $2-1/k>7/4$.

Put $z=k-q$.  We have $1\le z\le4$.  The two equal outer factors form a
factor $UVU$ with $|U|=q$ and $|V|=z$.  This factor has length at most $N$;
at the endpoint $2k=N+1$, its length is $2k-z\le N$.  Its exponent is at
least $(2q+z)/(q+z)=1+q/(q+z)$.  The $7/4^+$ condition therefore forces
$q\le3z\le12$.  Hence
 $k=q+z\le16$ and $p\le D=18k\le288$.  If instead $k\le4$, then directly
$p\le D=18k\le72<288$, so the same conclusion holds.
\end{proof}

For $r\in\{1,\ldots,17\}$ define
\begin{equation}
 g_r=\begin{cases}
 r+18,&r\in\{7,9\},\\
 r,&\text{otherwise}.
 \end{cases}
 \label{eq:canonical-circle-gaps}
\end{equation}

\begin{lemma}[Finite boundary lemma]\label{lem:finite-circle-boundary}
Let $abcd$ be a square-free ternary word of length four and let
$r\in\{1,\ldots,17\}$.  In the outer context
\[
 d_{19}(a)d_{19}(b)d_{23}(c)d_{23}(d)
\]
one can fix the $84$ Brinkhuis choices periodically, with period at most
three, and delete a gap of length $g_r$ beginning within $86$ output
characters of the central $d_{19}$--$d_{23}$ boundary so that no square of
root at most $288$ crosses the new seam.  The period and the position of the
gap depend only on $r$, not on $abcd$.
\end{lemma}

\begin{proof}
Number the $84$ letters of the displayed outer context from left to right,
and let the central boundary have output coordinate zero.  In the following
table, repeat the binary word $\pi_r$ periodically across those $84$ letters,
starting at the first one, and begin the deleted gap at coordinate $h_r$.
\[
\begin{array}{c|c|r|c@{\qquad}c|c|r|c}
r&g_r&h_r&\pi_r&r&g_r&h_r&\pi_r\\ \hline
1&1&-66&0   &10&10& 86&01\\
2&2& 86&10  &11&11& 86&0\\
3&3&-66&1   &12&12&-61&0\\
4&4&-61&01  &13&13&-66&0\\
5&5&-61&1   &14&14&-61&10\\
6&6& 86&01  &15&15&-61&01\\
7&25&86&100 &16&16&-66&1\\
8&8& 86&0   &17&17&-66&1\\
9&27&86&100 &&&&
\end{array}
\]
There are $18$ square-free ternary words $abcd$.  For each of the $17$ rows
and each of these $18$ contexts, direct comparison of the two halves, for
every root $1\le p\le288$ and every start whose square interior contains the
seam, finds no equality.  Thus the table requires $17\cdot18=306$
template--context cases, using $17$ periodic choice templates and three gap
coordinates.
The shortest left and right context buffers in these checks are $618$ and
$715$ characters, respectively, both larger than
$2\cdot288$; choices outside the displayed context therefore cannot affect
the conclusion.
\end{proof}

\begin{proposition}[Square-free defect source]\label{prop:square-free-defect-source}
For every $N>555$ and every nonzero residue $r$ modulo $18$, there is a
partially fixed occurrence-wise source whose outputs become circularly
square-free after deleting $g_r$ characters.
\end{proposition}

\begin{proof}
Choose the mixed outer circle in \cref{eq:dejean-mixed-form}.  If $a,b$
are the last two letters of $u$ and $c,d$ the first two letters of $v$, then
$abcd$ is a square-free factor of the circular word $(uv)$.  Fix the choices
and the gap given by \cref{lem:finite-circle-boundary}; leave all other
choices arbitrary.

Before deletion, every output is circularly square-free by
\cref{lem:circular-preservation}.  A square appearing after deletion must
therefore cross the new seam.  Lemma~\ref{lem:interior-seam-exclusion}
excludes the interior geometry, \cref{lem:dejean-edge-bound} bounds every
remaining root by $288$, and \cref{lem:finite-circle-boundary} handles those
roots.
\end{proof}

\subsection{The local lemma with a bounded defect}

Fix the source of \cref{prop:square-free-defect-source}; its remaining
choice variables are independent unbiased bits.  Put $\eta=1/2600$ and
$L=2600$.
For every cyclic start $t$ and $L<\ell\le n/2$, let $E_{t,\ell}$ be the
event that the two adjacent cyclic length-$\ell$ arcs $X,Y$ satisfy
$\LCS(X,Y)>(1-\eta)\ell$.

\begin{lemma}[Rational bounds for the defect calculation]
\label{lem:defect-rational-bounds}
Put $\alpha=(\log2+2\log(2600e)+\log3)/2600$.  Then
$\log2>6931471805/10^{10}$ and $\alpha<750704/10^8$.
Moreover, if $q=e^{-1/150}$ and $N=2601$, then
\[
 3\left(
   \frac{2q^N}{1-q}
   +\frac{q^{N+1}}{N(1-q)^2}
 \right)<\frac1{30000}.
\]
\end{lemma}

\begin{proof}
For $x>1$, put $z=(x-1)/(x+1)$.  Integrating the geometric series for
$2/(1-z^2)$ gives
\[
 \log x
 =2\sum_{j=0}^{M}\frac{z^{2j+1}}{2j+1}+R_M,
 \qquad
 0<R_M<
 \frac{2z^{2M+3}}{(2M+3)(1-z^2)}.
\]
Apply this estimate to $x=2,3,325/256$, with $M=12,18,6$,
respectively.  Since
$\log2600=11\log2+\log(325/256)$, direct rational substitution gives
the first two displayed bounds.

The upper estimate for $\log2$ from the same $M=12$ calculation gives
$25\log2<2601/150$ and hence $q^N<2^{-25}$.  Also
$e^{1/150}>1+1/150$ gives $q<150/151$, so $1-q>1/151$.  Therefore
\[
 3\left(
   \frac{2q^N}{1-q}
   +\frac{q^{N+1}}{N(1-q)^2}
 \right)
 <
 \frac3{2^{25}}\left(302+\frac{151^2}{2601}\right)
 =\frac{808303}{29091692544}
 <\frac1{30000}.
\]
\end{proof}

\begin{lemma}[Defect event tail]\label{lem:defect-event-tail}
Every such event satisfies $\Pr[E_{t,\ell}]<e^{-\ell/135}$.
\end{lemma}

\begin{proof}
A cyclic character arc of length $\ell$ contains at least
$\ell/18-4$ complete undeleted macroblocks.  There are at most two partial
blocks at its endpoints and at most two surviving macroblock fragments at
the defect seam.  At most $84$ complete blocks have fixed choices, so the
first arc contains at least $s\ge \ell/18-88$ complete macroblocks with free
choices.  Condition on every other free
choice.  This fixes $Y$, while any value of $X$ has conditional probability
at most $2^{-s}$.

By \cref{lem:lcs-ball} and the same binomial estimate used in
\cref{lem:ternary-volume}, the number of bad values of $X$ is at most
$e^{\alpha\ell}$, where
$\alpha=(\log2+2\log(2600e)+\log3)/2600<750704/10^8$.  It follows that
\[
 \Pr[E_{t,\ell}]
 \le 2^{88}\exp\left[-\left(\frac{\log2}{18}-\alpha\right)\ell\right].
\]
For every integer $\ell\ge2601$, the logarithmic margin against
$e^{-\ell/135}$ is at least its value at $2601$, namely
\[
 \left(\frac{\log2}{18}-\alpha-\frac1{135}\right)2601
 -88\log2
 >\frac{4444055899}{12000000000}
>\frac{3703}{10000}.
\]
Indeed, after collecting the $\log2$ terms, the left side is
$(113/2)\log2-2601\alpha-2601/135$, and the first comparison is the direct
substitution of
\cref{lem:defect-rational-bounds}.  This also shows that the coefficient of
$\ell$ is positive, so the margin is increasing.
\end{proof}

Assign charge $x(E_{t,\ell})=e^{-\ell/150}$.  As in
\cref{lem:neighbor-count}, an event of scale $\ell$ has at most
$3(\ell+r)$ neighbors of scale $r$.  Removing or fixing variables can only
delete dependencies.  Hence
\[
 \sum_{F\sim E_{t,\ell}}x(F)
 \le3\sum_{r\ge2601}(\ell+r)e^{-r/150}<\frac{\ell}{30000}.
\]
Indeed, after division by $\ell$ the left side is decreasing.  At
$\ell=N=2601$, the geometric-series identities give the left side divided
by $\ell$ as the expression in \cref{lem:defect-rational-bounds}, which is
below $1/30000$.

All charges are below $1/2$, so the neighboring product is greater than
$e^{-\ell/15000}$.  Since $1/135>1/150+1/15000$,
\cref{lem:defect-event-tail} satisfies the asymmetric local-lemma
criterion, and there is an assignment avoiding every event with
$\ell>2600$.

\begin{proof}[Proof of \cref{thm:all-length-circles}]
The chosen output is circularly square-free by
\cref{prop:square-free-defect-source}.  Therefore, for
$\ell\le2600$, adjacent cyclic arcs of length $\ell$ are unequal and
have LCS at most $\ell-1\le(1-1/2600)\ell$.  The local lemma gives the same
inequality for every larger equal scale up to
half the circumference.  The equal-to-arbitrary reduction
\cref{lem:equal-to-full} yields non-strict normalized edit distance at
least $(1/2600)/(2+1/2600)=1/5201$ in every rotation.  The strictly smaller
gap $1/5202$ gives parameter
$5201/5202$.

It remains to realize each length.  If $n\equiv0\pmod {18}$, use
\cref{cor:ternary-circles}: its strict gap $1/2002$ also implies the smaller
strict gap $1/5202$.  Otherwise put $r\equiv-n\pmod {18}$ with
$1\le r\le17$, choose $g_r$ from \cref{eq:canonical-circle-gaps}, and set
$N=(n+g_r)/18$.  For $n\ge10008$ we have $N>555$.  The construction above has length
$18N-g_r=n$, completing the proof.
\end{proof}

%% file: sections/05_quantitative_strengthening.tex
\section{Quantitative strengthening from the 54-uniform family}
\label{sec:quantitative-strengthening}

The six-codeword proof isolates the mechanism, while the larger 54-uniform
Brinkhuis family supplies substantially more local entropy.  We use its
$952$ occurrence-wise choices to control unequal adjacent intervals directly.
The calculation also retains the random macroblock fragment crossing the
left endpoint of the first interval.  The resulting suffix multiplicities and
numerical estimates appear in Appendix~\ref{app:numerics}.

The square-free property of the family is the published result below.
The additional finite statements used here are verified by integer and exact
rational computations in the artifact cited after
\cref{thm:quantitative-intro}.

\begin{proposition}[A 54-uniform Brinkhuis family]
\label{prop:large-codebook}
There is a special 54-Brinkhuis 952-triple
$\mathcal D_0,\mathcal D_1,\mathcal D_2\subseteq\ternary^{54}$ in which
$\mathcal D_0$ is closed under reversal, and $\mathcal D_1$ and
$\mathcal D_2$ are obtained from $\mathcal D_0$ by one and two cyclic
permutations of the alphabet.  Consequently, for every ternary square-free
outer word $T=T_1\cdots T_N$ and every occurrence-wise choice vector
$R\in[952]^N$, the word $D_{T_1,R_1}\cdots D_{T_N,R_N}$ is square-free,
and the map from choice vectors to outputs is injective.
\end{proposition}

\begin{proof}
Sollami, Douglas, and Liebmann constructed this special 54-Brinkhuis
952-triple satisfying the length-three condition in their definition
\cite[Definition~1 and Theorem~4]{SollamiDouglasLiebmann2016}.  In their displayed
construction, the last 476 words of $\mathcal D_0$ are the reversals of the
first 476, so $\mathcal D_0$ is reversal closed.  Independent
occurrence-wise substitution then preserves square-freeness, as in
\cref{prop:codebook}.  Injectivity follows because the 952 codewords for a
fixed outer letter are distinct and the macroblock boundaries are fixed.
\end{proof}

Fix $m=54$ and $K=952$, and choose independently and uniformly one codeword
at every occurrence of a fixed ternary square-free outer word.  Write $W$
for the resulting word.  Every outcome is square-free by
\cref{prop:large-codebook}.  We index positions within a codeword from $0$.

For $1\le s\le53$, let
\[
 M_s=\max_z\bigl|\{r\in[952]:
       D_{a,r}[54-s,54)=z\}\bigr|.
\]
This quantity is independent of $a$: the three codebooks differ only by
cyclic permutations of the alphabet. Its exact values are given in
\cref{lem:suffix-fibers}.

For an integer $\ell\ge54$ and an offset $a\in\{0,\ldots,53\}$, define
\begin{equation}\label{eq:fragment-point-mass}
 P_{\ell,a}=
 \begin{cases}
  952^{-\lfloor\ell/54\rfloor},&a=0,\\[2mm]
  M_{54-a}\,952^{-1-\lfloor(\ell-54+a)/54\rfloor},&a>0,
 \end{cases}
 \qquad
 P_\ell=\max_{0\le a<54}P_{\ell,a}.
\end{equation}

\begin{lemma}[Fragment-aware conditional mass]
\label{lem:fragment-conditional-mass}
Let $X=W[t,t+\ell)$, where $\ell\ge54$, and let $Y$ be any nonempty
interval immediately to the right of $X$.
One can expose a set of occurrence choices that fixes $Y$ and, for every
resulting conditioning of positive probability, gives every exact value of
$X$ conditional probability at most $P_\ell$.
\end{lemma}

\begin{proof}
Let $a$ be the offset of $t$ in its macroblock. If $a=0$, leave
unconditioned all $\lfloor\ell/54\rfloor$ macroblocks contained in $X$ and
condition on every other occurrence choice. None of these blocks meets
$Y$, so $Y$ is fixed. Distinctness of the $952$ codewords implies that an
exact value of $X$ determines each unconditioned choice. Its conditional
probability is therefore at most $P_{\ell,0}$.

Suppose $a>0$ and put $s=54-a$. The macroblock containing the first
character of $X$ contributes a suffix of length $s$. Because $\ell\ge54>s$,
this block ends strictly before $Y$ begins. Leave its choice and all later
macroblocks wholly contained in $X$ unconditioned, and condition on every
other choice. Again $Y$ is fixed. There are
$\lfloor(\ell-s)/54\rfloor$ later complete blocks. The visible suffix is
compatible with at most $M_s$ choices, while the value of every complete
block determines its choice. Independence and uniformity therefore give
the upper bound $M_s/952^{1+\lfloor(\ell-s)/54\rfloor}=P_{\ell,a}$.
Taking the maximum over the possible offset proves the claim.
\end{proof}

\input{appendix/F_supported_alignments}

\subsection{Direct control of unequal intervals}

\linespread{1.055}\selectfont
\setlength{\parskip}{2.5pt plus 0.5pt minus 0.25pt}

To avoid the loss in the equal-to-unequal reduction, we place unequal
interval events directly in the local lemma.  This calculation becomes effective once the intervals expose
enough macroblock entropy.  Before that point, parity and square-freeness
leave only finite families of one- and two-deletion obstructions.  The
following exact properties of the codebook remove both families.

\begin{lemma}[The bounded one-deletion shells are absent]
\label{lem:shortest-shell}
Let $W$ be any occurrence-wise output of the 54-uniform 952-family applied to
a ternary square-free outer word.  Let $s$ be odd with $229\le s\le463$.
If adjacent factors $X,Y$ have lengths $\{(s-1)/2,(s+1)/2\}$, then
$\ED(X,Y)>1$.
\end{lemma}

\begin{proof}
Fix $s$ and put $r=(s+1)/2$.  If $\ED(X,Y)=1$, deleting one position from the
length-$r$ factor gives the length-$(r-1)$ factor.  Fix the orientation, the
offset of the combined length-$s$ factor in its first macroblock, and the
deleted position.  Equality of the surviving words gives $r-1$ coordinate
equalities between the macroblocks met by the factor.  Up to a cyclic
permutation of the alphabet, it remains only to enumerate the square-free
outer contexts meeting that factor.

Group the coordinate equalities according to the pair of macroblocks that
they compare.  For one group and one cyclic difference $\delta$ between its
two outer letters, record whether the exact relation on the two 952-element
choice sets is empty.  The resulting three-bit labels form a finite
\emph{mask graph} on the met macroblocks.

\Cref{lem:shortest-shell-verification} shows that these graphs admit no
square-free outer coloring except for patterns at $s=251$.  In each exception,
the graph forces an outer context of the form $abcab$ with $a,b,c$ distinct,
and the corresponding codeword relations are inconsistent.  Thus no
one-deletion equality is possible.  The exact relations, exhaustive
classification, and residual consistency checks are given in that lemma.
\end{proof}

\begin{lemma}[The bounded two-deletion shells are absent]
\label{lem:two-deletion-shell}
Let $W$ be as in \cref{lem:shortest-shell}.  If adjacent factors $X,Y$
satisfy $|X|+|Y|\in\{456,462,464\}$ and
$\lvert |X|-|Y|\rvert\le2$, then $\ED(X,Y)>2$.
\end{lemma}

\begin{proof}
If the lengths differ by two, edit distance two would delete two positions
from the longer factor and recover the shorter one.  If the lengths are equal,
square-freeness excludes edit distance zero, while edit distance two would
delete one position from each factor and leave equal subsequences.  These
are the only possible violations, since edit distance has the same parity
as the total length.

For each mode, fix the offset in the first macroblock and the two deleted
positions, and group the surviving coordinate equalities by their ordered
pair of macroblocks, exactly as in \cref{lem:shortest-shell}.
\Cref{lem:two-deletion-shell-verification} checks all three modes at each
total.  For totals 456 and 462, every pattern contains a coordinate group
incompatible with every cyclic outer-letter difference.  At total 464,
patterns without such a conflict have mask graphs with no square-free outer
coloring.  In either case the full equality system is inconsistent, proving
the claim.
\end{proof}

Put $D=216$ and
$\mathcal H=\{229,231,\ldots,463\}\cup\{456,462,464\}$.
For positive integers $a,b$ with total length $s=a+b\ge217$,
$s\notin\mathcal H$, and $|a-b|<s/D$, define
\[
 \widehat E_{t,a,b}
 =\left\{
   \ED\bigl(W[t,t+a),W[t+a,t+a+b)\bigr)<\frac{s}{D}
  \right\}.
\]
For such a pair, let $\mathcal J_{a,b}$ consist of the nonzero integer pairs
$(u,v)$ with $0\le u<a$, $0\le v<b$, $v-u=b-a$,
$D(u+v)<a+b$, and $a-u\le F_a(u,v)$.  Set
\begin{equation}\label{eq:unequal-deletion-volume}
 \begin{aligned}
 \widehat V^{\rm L}_{a,b}
 &=\sum_{(u,v)\in\mathcal J_{a,b}}
    \binom au\binom bv2^u,\\
 \widehat V^{\rm R}_{a,b}
 &=\sum_{(u,v)\in\mathcal J_{a,b}}
    \binom au\binom bv2^v,\\
 \widehat p_{a,b}
 &=\min\{P_a\widehat V^{\rm L}_{a,b},
          P_b\widehat V^{\rm R}_{a,b}\}.
 \end{aligned}
\end{equation}

\begin{lemma}[Direct unequal-event bound]
\label{lem:unequal-event-bound}
For every admissible $t,a,b$ as above,
$\Pr[\widehat E_{t,a,b}]\le\widehat p_{a,b}$.
\end{lemma}

\begin{proof}
The imbalance condition implies
$a>(s-s/D)/2\ge217\cdot215/(2\cdot216)>54$, so
\cref{lem:fragment-conditional-mass} applies to the left interval and
fixes the right interval while giving point mass at most $P_a$.

For any bad pair, put $c=\LCS(X,Y)$, $u=a-c$, and $v=b-c$.  Then
$v-u=b-a$ and $u+v=\ED(X,Y)<s/D$.  This implies $c>0$, and
\cref{lem:alignment-capacity} gives $c\le F_a(u,v)$.
For fixed $u,v$, \cref{lem:supported-fillings} bounds the possible values of
$X$ by $\binom au\binom bv2^u$.
The case $u=v=0$ would give the square $XX$ and has probability zero because
every outcome is square-free. Summing the remaining cases and multiplying by
the conditional point-mass bound gives
$\Pr[\widehat E_{t,a,b}]\le P_a\widehat V^{\rm L}_{a,b}$ after averaging over
the conditioning.  Reverse the entire output and interchange the two
intervals.  The family is closed under reversal; moreover,
reversal merely permutes the uniformly chosen codewords within each
outer-letter family.  The reversed outer word is square-free, so the
reversed random output has exactly the same form as the original one and
the same conditional-mass and supported-filling argument gives
$\Pr[\widehat E_{t,a,b}]\le P_b\widehat V^{\rm R}_{a,b}$: the original
capacity restriction is necessary for any realization of the pair,
regardless of which interval is exposed first.  Taking the smaller
of the two valid bounds proves the claim.
\end{proof}

If $\widehat p_{a,b}=0$, the preceding lemma gives probability zero.  Since
the finite product space is uniform and every choice vector has positive
mass, $\widehat E_{t,a,b}$ is empty.  Discard these events and let
$\mathcal I$ denote the remaining ordered length types: those with
$a,b\ge1$, $a+b\ge217$, $a+b\notin\mathcal H$,
$216|a-b|<a+b$, and $\widehat p_{a,b}>0$.

Assign the charge
\begin{equation}\label{eq:unequal-charges}
 \widehat x_{a,b}
 =\widehat p_{a,b}
   \exp\!\left(\frac{37(a+b)}{100000}+\frac{21}{100}\right).
\end{equation}

\begin{lemma}[Unequal-event charge envelopes]
\label{lem:unequal-charge-envelopes}
For $j\in\{0,1\}$, put
$\widehat R_j=\sum_{(a,b)\in\mathcal I}(a+b)^j
\frac{\widehat x_{a,b}}{1-\widehat x_{a,b}}$.
Every retained charge is smaller than $1/6000$, and
\begin{align*}
 \widehat R_0
 &<\frac{364756}{10^9}<\frac{37}{100000},\\
 109\widehat R_0+\widehat R_1
 &<\frac{203468108}{10^9}<\frac{21}{100}.
\end{align*}
\end{lemma}

\begin{proof}
This is the exact finite calculation and analytic tail proved in
\cref{lem:unequal-charge-bounds}.
\end{proof}

\begin{proposition}[Direct normalized edit-distance gap]
\label{prop:direct-unequal-gap}
For every $n$ there is a word $S\in\ternary^n$ such that every pair of
adjacent nonempty intervals $A,B$ satisfies
\[
  \ED(A,B)\ge\frac{|A|+|B|}{216}.
\]
\end{proposition}

\begin{proof}
Fix the requested length $n$, put $N=\lceil n/54\rceil$, and choose a
ternary square-free outer word of length $N$.  Apply the 54-uniform family to
obtain a random word $W$ of length $54N$ as above.

Declare that $\widehat E_{t,a,b}$ depends on every macroblock variable
meeting its support of length $s=a+b$. A fixed event has at most
$s+r+109$ neighbors of any fixed ordered length type $(c,d)$ of total
$r=c+d$, by the same support-overlap argument as
\cref{lem:general-neighbor-count}. Thus
\begin{align*}
 -\log\prod_{F\sim\widehat E_{t,a,b}}(1-x_F)
 &\le\sum_{c,d}(s+c+d+109)
          \frac{\widehat x_{c,d}}{1-\widehat x_{c,d}}\\
 &=(s+109)\widehat R_0+\widehat R_1\\
   &<\frac{37s}{100000}+\frac{21}{100}.
\end{align*}
Thus the neighboring product is greater than
$\exp(-37s/100000-21/100)$.  Multiplying by the charge in
\cref{eq:unequal-charges} gives a quantity greater than
$\widehat p_{a,b}\ge\Pr[\widehat E_{t,a,b}]$.
The asymmetric local lemma therefore supplies an outcome avoiding every
retained direct bad event.

It remains to check pairs not represented by an event.  If $s\le216$, then
unequal lengths give
$\ED(A,B)\ge\lvert |A|-|B|\rvert\ge1\ge s/216$, while equal lengths give
$\ED(A,B)\ge2$ by square-freeness. Now let $s\ge217$. If
$s\notin\mathcal H$, the imbalance either gives the desired bound, the pair
is one of the avoided retained events, or $\widehat p_{a,b}=0$ and the direct
event is empty by \cref{lem:unequal-event-bound}.  If
$s\in\{229,231,\ldots,463\}$, failure of the
imbalance bound forces $|a-b|=1$; edit distance has odd parity, and a
violation of the $s/216$ bound must have edit distance one, contrary to
\cref{lem:shortest-shell}. At $s\in\{456,462,464\}$, failure of the imbalance
bound leaves differences zero and two. Square-freeness excludes edit
distance zero, while \cref{lem:two-deletion-shell} excludes edit distance
two.
Thus every adjacent pair in $W$ has the claimed non-strict bound.  Finally
take its length-$n$ prefix; every adjacent pair in the prefix is the same
pair in $W$, so the inequality is inherited.
\end{proof}

\begin{proof}[Proof of \cref{thm:quantitative-intro}]
Apply \cref{prop:direct-unequal-gap}.  The strict synchronization inequality
holds for every $\eps>215/216$; in particular one may take
$\eps=\maineps$.
\end{proof}

\linespread{1.03}\selectfont
\setlength{\parskip}{1.5pt plus 0.4pt minus 0.2pt}

\begin{corollary}[An infinite ternary synchronization string]
\label{cor:infinite-ternary}
There exists an infinite ternary word whose every finite adjacent interval
pair satisfies the non-strict normalized edit-distance lower bound
$\mainboundarygap$.  Consequently it is an $\eps$-synchronization string for
every $\eps>\mainthreshold$.
\end{corollary}

\begin{proof}
Let $\mathcal G_n$ be the set of ternary words of length $n$ satisfying the
stated bound for every adjacent interval pair.  Each $\mathcal G_n$ is
nonempty by \cref{prop:direct-unequal-gap}, and every prefix of a word in
$\mathcal G_n$ belongs to the corresponding earlier level.  The union of the
$\mathcal G_n$, ordered by the prefix relation, is therefore an infinite
finitely branching tree with a nonempty level at every depth.  K\"onig's
lemma supplies an infinite branch.  Every finite adjacent interval pair on
that branch lies in one of its finite prefixes and hence satisfies the weak
gap.  The strict synchronization statement follows because
$1-\eps<1/216$ whenever $\eps>215/216$.
\end{proof}

%% file: appendix/F_supported_alignments.tex
\subsection{Deletion balls within the source support}
\label{sec:supported-alignments}

Two restrictions on the support sharpen the ambient deletion-ball bound.
First, neighboring letters are unequal.  Second, the block phases restrict
the lengths of uninterrupted matching runs.  The first restriction reduces
the number of fillings; the second eliminates deletion patterns before we
count their fillings.

\begin{lemma}[Supported fillings]\label{lem:supported-fillings}
Fix a word $y$ of length $b$ over $q\ge2$ symbols, and integers
$0\le u<a$, $0\le v<b$ with $a-u=b-v>0$.  The number of length-$a$ words
$x$ with unequal neighboring letters that become equal to $y$ after deleting
$u$ and $v$ positions, respectively, is at most
$\binom au\binom bv(q-1)^u$.
\end{lemma}

\begin{proof}
Choose both deletion sets.  The undeleted letters of $x$ are then fixed.
Fill the initial run of deleted positions from right to left, starting at
the first undeleted letter.  Fill each later run from left to right,
starting at the preceding undeleted letter.  Every step has a previously
fixed neighbor and hence at most $q-1$ choices.  Requirements at the other
end of a run can only reduce this number.  There is an undeleted letter
because $a-u>0$, so the procedure covers all deleted positions.  Multiplying
by the numbers of deletion sets proves the bound; multiple descriptions of
the same word only overcount.
\end{proof}

For the 54-uniform family, define $B_d$ for $1\le d\le27$ by the table
below, and put $B_{54-d}=B_d$.  An infinite entry means that we impose no
restriction at that phase.
\begin{center}
\small
\begin{tabular}{c|rrrrrrrrr}
$d$ &1&2&3&4&5&6&7&8&9\\
$B_d$&41&41&41&41&47&45&83&36&39\\[2pt]
$d$&10&11&12&13&14&15&16&17&18\\
$B_d$&42&49&56&50&44&47&44&$\infty$&$\infty$\\[2pt]
$d$&19&20&21&22&23&24&25&26&27\\
$B_d$&61&72&39&66&75&82&$\infty$&46&52
\end{tabular}
\end{center}
For a positive absolute lag $\lambda$, set
\begin{equation}\label{eq:run-bound}
 G(\lambda)=
 \begin{cases}
 \max\{42,\lambda-54\},&54\mid\lambda,\\
 B_{\lambda\bmod54},&54\nmid\lambda.
 \end{cases}
\end{equation}

\begin{lemma}[Matching runs]\label{lem:phase-runs}
In any output $W$ of the family in \cref{prop:large-codebook}, two equal
factors whose starting positions differ by $\lambda>0$ have length at most
$G(\lambda)$.
\end{lemma}

\begin{proof}
We first treat zero phase without enumerating outer contexts.  Every word
in $\mathcal D_0$ has prefix $012021$ and suffix $120210$; the other two
codebooks are its cyclic shifts.  Thus a letter at any of the first or last
six coordinates identifies the outer letter when its coordinate is known.
These common prefix and suffix properties are checked directly on the 952
listed words.

Write $\lambda=54d$.  If $d=1$, a matching run containing any such fixed
coordinate would identify two consecutive outer letters, contradicting
square-freeness.  A run can therefore occupy only the 42 interior
coordinates of one block.  If $d\ge2$, a run of length $54(d-1)+1$
contains a fixed coordinate in each of $d$ consecutive blocks.  Indeed, a
run starting at coordinate zero contains the first coordinate in all $d$
blocks; a run starting at a positive coordinate contains the last coordinate
of its first block and the first coordinate of each of the next $d-1$
blocks.  Equality at lag $54d$ identifies those $d$ outer letters with the
next $d$ outer letters.  This gives a square in the outer word, a
contradiction.  The zero-phase bound follows.

For a nonzero phase $d$ with $B_d<\infty$, fix the starting coordinate
$p\in\{0,\ldots,53\}$ of a hypothetical run of length $B_d+1$.
Split the run at every block boundary on either side.  Each resulting
piece compares coordinates $[i,i+r)$ of one codeword with coordinates
$[j,j+r)$ of another.  Form the two finite sets
$\{D[i,i+r):D\in\mathcal D_0\}$ and
$\{D[j,j+r):D\in\mathcal D_0\cup\mathcal D_1\cup\mathcal D_2\}$.
After a common cyclic relabeling, every possible comparison belongs to
these sets.  Direct enumeration finds a piece with disjoint sets for every
one of the 54 starts at every finite phase in the table.  This is a finite
test of $47\cdot54=2538$ phase--start pairs, involving 2756 distinct
coordinate-set comparisons.  It uses no assumption about choices or outer
letters on different blocks: allowing all of them independently only
enlarges the two sets.  Disjointness rules out the hypothetical run.
This proves every finite entry and hence the claim.
\end{proof}

Define $F_a(i,j)$ for $0\le i<a$, $0\le j<b$ by
$F_a(0,0)=G(a)$ and, for $i+j>0$,
\begin{equation}\label{eq:alignment-capacity}
 F_a(i,j)=G(a+j-i)+
 \max\{F_a(i-1,j),F_a(i,j-1)\},
\end{equation}
omitting terms with a negative index.  Arithmetic with $\infty$ has its
usual extended-real meaning.  All lags in this recurrence are positive
because $i<a$.

\begin{lemma}[Deletion-pattern capacity]\label{lem:alignment-capacity}
Suppose adjacent factors of lengths $a,b$ have a common subsequence of
length $c=a-u=b-v>0$.  Then $c\le F_a(u,v)$.
\end{lemma}

\begin{proof}
Trace the alignment defining that subsequence.  Between consecutive
deletions, the matched letters form an uninterrupted run.  After $i$
deletions on the left and $j$ on the right, their absolute lag is
$a+j-i$, so \cref{lem:phase-runs} bounds the run by $G(a+j-i)$.
The deletion steps describe a monotone path from $(0,0)$ to $(u,v)$.
The recurrence maximizes the sum of these run bounds over all such paths.
Allowing a run of its full upper-bound length at every visited vertex can
only increase the number of matches.  It therefore bounds their actual
total $c$.
\end{proof}

For example, at total length 227 a one-deletion alignment would need 113
matches, whereas the two runs have total capacity at most
$G(113)+G(114)=47+45=92$.  At total length 452, an edit-distance-two
alignment has at most three runs, all at lags between 225 and 227.
Their phase bounds are at most 49, giving at most 147 matches rather than
the required 225.  These exclusions follow from the same recurrence used
below; neither requires a separate deletion-shell enumeration.

%% file: appendix/B_numeric_details.tex
\section{Exact calculations for the ternary parameters}
\label{app:numerics}

This appendix proves the analytic inequalities and finite reductions used in
\cref{sec:ternary-proof,sec:quantitative-strengthening}.  The finite claims
below concern the published 54-uniform codebook from
\cref{prop:large-codebook}; every comparison is expressed through the exact
relations defined here.

\begin{lemma}[Elementary logarithmic bounds]\label{lem:elementary-logs}
The following strict inequalities hold:
\[
  \frac{69}{100}<\log2<\frac7{10},\qquad
  \log3<\frac{11}{10},\qquad
  \log10<\frac73,\qquad
  \log4<2.
\]
\end{lemma}

\begin{proof}
Finite lower truncations of the exponential series give
\begin{align*}
  e^{7/10}
  &>\frac{12013}{6000}>2,\\
  e^{11/10}
  &>\frac{36015101}{12000000}>3,\\
  e^{7/3}
  &>\frac{1071641}{104976}>10,
  \qquad e^2>1+2+2^2/2>4.
\end{align*}
These imply the four upper bounds.  For the lower bound, the elementary
estimate $j!\ge2\cdot3^{j-2}$ for $j\ge2$ yields, with $x=69/100$,
\[
  e^x\le1+x+\frac{3x^2}{2(3-x)}
  =\frac{92361}{46200}<2,
\]
so $69/100<\log2$.
\end{proof}

\begin{lemma}[Logarithms for the 54-uniform family]
\label{lem:large-family-log-bounds}
The logarithm satisfies $137/20<\log952<7$.
\end{lemma}

\begin{proof}
The positive exponential series gives
$e>\sum_{j=0}^{7}1/j!=685/252$.
For an upper bound on $e$, use $(7+k)!\ge7!\,8^k$ for every $k\ge0$ to obtain
\[
 e
 =\sum_{j=0}^{6}\frac1{j!}+\sum_{j=7}^{\infty}\frac1{j!}
 \le\frac{1957}{720}+\frac1{7!}\sum_{k=0}^{\infty}\frac1{8^k}
 =\frac{31967}{11760}<\frac{2719}{1000}.
\]

The lower bound on $e$ gives $e^7>(685/252)^7>952$, so $\log952<7$.
For the other direction,
\[
 e^{137/20}=\frac{e^7}{e^{3/20}}
 <\frac{(2719/1000)^7}
 {1+3/20+(3/20)^2/2}
 =\frac{(2719/1000)^7}{929/800}<952.
\]
Taking logarithms proves the claim.
\end{proof}

\begin{lemma}[Suffix fibers of the 54-uniform family]
\label{lem:suffix-fibers}
For the special 54-Brinkhuis 952-triple in
\cref{prop:large-codebook}, the suffix multiplicities $M_s$ defined in
\cref{sec:quantitative-strengthening} are as follows:
\begingroup\normalfont
\begin{center}
\begin{tabular}{@{}r|*{14}{r}@{}}
$s$&1&2&3&4&5&6&7&8&9&10&11&12&13&\\
$M_s$&952&952&952&952&952&952&642&642&642&444&444&233&211&\\[3pt]
$s$&14&15&16&17&18&19&20&21&22&23&24&25&26&\\
$M_s$&211&198&132&100&100&85&85&66&37&37&33&26&25&\\[3pt]
$s$&27&28&29&30&31&32&33&34&35&36&37&38&39&\\
$M_s$&22&15&15&13&10&9&9&7&7&6&5&5&4&\\[3pt]
$s$&40&41&42&43&44&45&46&47&48&49&50&51&52&53\\
$M_s$&3&3&3&2&2&1&1&1&1&1&1&1&1&1
\end{tabular}
\end{center}
\endgroup
\end{lemma}

\begin{proof}
For each $s$, partition the $952$ words in $\mathcal D_0$ by
their suffix of length $s$ and take the largest class.  Direct enumeration
gives the displayed row.  The other two codebooks are obtained from
$\mathcal D_0$ by bijective alphabet permutations, which preserve equality
of suffixes and hence preserve every fiber size.
\end{proof}

We will use the following elementary rational enclosure in the finite
calculation.  For every rational $0\le z<82$,
\begin{equation}\label{eq:rational-exp-upper}
 e^z
 \le\sum_{j=0}^{80}\frac{z^j}{j!}
  +\frac{z^{81}}{81!}\frac1{1-z/82}.
\end{equation}
Indeed, beginning with the term of degree $81$, the ratio of consecutive
terms is at most $z/82<1$, so the positive tail is bounded by the displayed
geometric series.  The inequality is strict when $z>0$.

\begin{lemma}[Finite verification of the bounded one-deletion shells]
\label{lem:shortest-shell-verification}
The conclusion of \cref{lem:shortest-shell} holds for the
54-uniform 952-family.
\end{lemma}

\begin{proof}
We reduce the claim to finite relations among codeword choices.  Fix an odd
$s\in[229,463]$ and put $r=(s+1)/2$.  If adjacent lengths $r$ and $r-1$ have
deletion distance one,
deleting one of the $r$ positions in the longer word makes the two surviving
length-$(r-1)$ words equal.  A case is specified by the orientation, one of
the 54 offsets of the combined factor in its first macroblock, the deleted
position, and the square-free outer context met by the factor.

Equality after the deletion imposes $r-1$ character equalities.  Group them
by the ordered pair $(i,j)$ of macroblocks containing the two characters, and
let $P_{i,j}$ be the resulting list of local-coordinate pairs.  For
$\delta\in\mathbb Z_3$, define the exact choice relation
\[
 \mathcal R_\delta(P_{i,j})=
 \left\{(u,v)\in[952]^2:
 D_{0,u}[p]=D_{\delta,v}[q]
 \text{ for every }(p,q)\in P_{i,j}\right\}.
\]
The three codebooks are cyclic alphabet translates, so if the two outer
letters differ by $\delta$, this is precisely the relation that their two
occurrence choices must satisfy.  Label edge $(i,j)$ by the three-bit mask
\[
 \mu(P_{i,j})=\{\delta\in\mathbb Z_3:
                 \mathcal R_\delta(P_{i,j})=\varnothing\}.
\]
For the fixed orientation, offset, and deletion position, these labels form
a mask graph on the met outer occurrences.  An outer context can support the
putative equality only if its difference on every edge avoids that edge's
mask.

There are
\[
 \sum_{\substack{229\le s\le463\\s\text{ odd}}}
   2\cdot54\cdot\frac{s+1}{2}
 =2{,}211{,}084
\]
orientation/offset/deletion patterns.  Exact finite-set comparison constructs
13,288 distinct mask graphs.  Recursively generating normalized square-free
colorings of each graph shows that every graph has no such coloring except for
280 patterns at total $s=251$.  These exceptions belong to 12 mask-graph types.
Each contains edges $(0,3)$ and $(1,4)$ with mask $\{1,2\}$, forcing
$T_0=T_3$ and $T_1=T_4$.  The factor meets five macroblocks, and
square-freeness now forces the normalized outer context to be $01201$ or
$02102$.  Across the 280 patterns this gives exactly 344 residual contexts.

For a residual context, give each occurrence the domain $[952]$ and put the
exact relation $\mathcal R_{T_j-T_i}(P_{i,j})$ on every edge.  Repeatedly
delete a choice having no support across an incident relation.  This
arc-consistency operation preserves every satisfying assignment.  A minimal
edge subsystem is already inconsistent: 338 residual contexts have a
two-edge core and the remaining six have a three-edge core.  Removing
redundant coordinate equalities from these relations leaves 48 core types,
each involving at most 14 character equalities, and every core still empties
a domain.

These cores cover all 344 residual contexts arising from the mask-graph
classification.  Hence no residual context admits a satisfying assignment,
and every one-deletion shell in the stated range is empty.
\end{proof}

\begin{lemma}[Finite verification of the bounded two-deletion shells]
\label{lem:two-deletion-shell-verification}
The conclusion of \cref{lem:two-deletion-shell} holds for the
54-uniform 952-family.
\end{lemma}

\begin{proof}
We first enumerate the three modes at total 456.  Their lengths are
$(229,227)$, $(227,229)$, and $(228,228)$.  Each imbalanced mode has
$54\binom{229}{2}=1{,}409{,}724$ normalized offset/deletion patterns;
the equal-length mode has $54\cdot228^2=2{,}807{,}136$.
Every one of these $5{,}626{,}584$ patterns contains a coordinate group
$P_{i,j}$ for which the exact relation $\mathcal R_\delta(P_{i,j})$ from
\cref{lem:shortest-shell-verification} is empty for all three cyclic
outer-letter differences $\delta$.  Hence every pattern is impossible
independently of the outer context.

At total 462 the three modes have lengths
$(232,230)$, $(230,232)$, and $(231,231)$.  Each imbalanced mode has
$54\binom{232}{2}=1{,}446{,}984$ patterns, and the equal-length mode has
$54\cdot231^2=2{,}881{,}494$.
For every one of these $5{,}775{,}462$ patterns, one coordinate group
$P_{i,j}$ has $\mathcal R_\delta(P_{i,j})=\varnothing$ for all three possible
cyclic outer-letter differences.  Thus every pattern is impossible
independently of the outer context.

At total 464, for lengths
$(233,231)$, choose the two deleted positions in the left factor; for
$(231,233)$ choose them in the right factor; for $(232,232)$ choose one
deleted position in each factor. Together with the 54 offsets in the first
macroblock, this gives $54\binom{233}{2}=1{,}459{,}512$ patterns for each
imbalanced mode and $54\cdot232^2=2{,}906{,}496$ for the equal-length mode,
for a total of $5{,}825{,}520$.

For each pattern, group the surviving coordinate equalities by their ordered
pair of macroblocks and form the relations
$\mathcal R_\delta(P_{i,j})$ from the preceding proof. In $1{,}456{,}707$,
$1{,}456{,}707$, and $2{,}899{,}660$ patterns, respectively, one group has
an empty exact relation for all three possible cyclic differences between
the two outer letters. These $5{,}813{,}074$ patterns are therefore impossible for every
outer context.

For each remaining pattern form the three-bit mask graph from the preceding
proof.  The $2{,}805$ left-longer patterns collapse to 52 mask graphs, the
$2{,}805$ right-longer patterns to 52, and the $6{,}836$ equal-length
patterns to 112.  Exact generation of normalized square-free outer colorings
shows that none of these 216 mode-specific graphs has a surviving coloring.
Thus all $12{,}446$ residual patterns are inconsistent without any search
over codeword choices.

The half-family contains 476 words; adjoining their reversals gives
952 distinct words.  Forming the exact relations for these finite
codebooks gives the counts above at all three totals.  The empty-relation
tests and the square-free coloring tests together exclude every pattern.
Hence no edit-distance-two pair exists in any stated shell.
\end{proof}

\begin{lemma}[Exact bounds for the direct unequal events]
\label{lem:unequal-charge-bounds}
The charges in \cref{eq:unequal-charges} satisfy
\[
 \widehat x_{a,b}<\frac1{6000},\qquad
 \widehat R_0<\frac{364756}{10^9}<\frac{37}{100000},\qquad
 109\widehat R_0+\widehat R_1
 <\frac{203468108}{10^9}<\frac{21}{100}.
\]
\end{lemma}

\begin{proof}
We again separate an exact finite calculation from a uniform tail. Total
lengths in $\mathcal H$ are absent by
\cref{lem:shortest-shell-verification,lem:two-deletion-shell-verification}.
For every $217\le s\le4000$ with $s\notin\mathcal H$, put
$z_s=37s/100000+21/100$.
Apply the rational exponential upper bound \cref{eq:rational-exp-upper} with
$z=z_s$; it is valid because $z_s<82$. For every ordered
 pair $a+b=s$ with $216|a-b|<s$, compute $P_a,P_b$ from
\cref{eq:fragment-point-mass,lem:suffix-fibers} and compute the two integer
volumes in \cref{eq:unequal-deletion-volume}. The term
$\min\{P_a\widehat V^{\rm L}_{a,b},P_b\widehat V^{\rm R}_{a,b}\}e^{z_s}$
is rounded upward to the next multiple of $10^{-30}$.  The volumes use the
capacity recurrence \cref{eq:alignment-capacity} and the supported filling
factors $2^u,2^v$.  There are $29{,}961$
nonzero terms in this range after the zero-volume types are discarded.
Exact integer addition gives
\begin{align}
 \max_{\substack{(a,b)\in\mathcal I\\a+b\le4000}}\widehat x_{a,b}
 &<\frac{158118}{10^9}<\frac1{6000},
 \label{eq:unequal-finite-max}\\
 \frac{6000}{5999}
 \sum_{\substack{(a,b)\in\mathcal I\\a+b\le4000}}\widehat x_{a,b}
 &<\frac{364755781}{10^{12}},
 \label{eq:unequal-finite-zero}\\
 \frac{6000}{5999}
 \sum_{\substack{(a,b)\in\mathcal I\\a+b\le4000}}
 (a+b)\widehat x_{a,b}
 &<\frac{163709727665}{10^{12}}.
 \label{eq:unequal-finite-one}
\end{align}
Here $\mathcal I$ is the set of retained ordered length types defined before
\cref{eq:unequal-charges}.  The factor
$6000/5999$ is valid because the first inequality implies
$x/(1-x)<6000x/5999$. As above, these are finite
comparisons of explicitly defined rational numbers.

For the tail, let $h=(\log952)/54$ and
$\nu=(\log2+\log(216e)+\log3)/216$.
If $s=a+b$ and $216|a-b|<s$, then
$a>215s/432$. Since
$\widehat p_{a,b}\le P_a\widehat V^{\rm L}_{a,b}$, it is enough in the tail
to bound this first conditioning direction. By
\cref{lem:general-conditional-mass}, the complete-block point mass is at
most $952^2e^{-ha}$. We next bound the unequal deletion volume. Put
$d=u+v$. Vandermonde's identity gives
$\binom au\binom bv\le\binom sd$, and $2^u\le3^d$.  The relation
$v-u=b-a$ determines at most one pair $(u,v)$ for a fixed $d$.  There are fewer than
$s/216+1\le2^{s/216}$ possible
integer values of $d<s/216$. Using
\cref{lem:binomial-estimates} and monotonicity at $d=s/216$ therefore gives
\[
 \widehat V^{\rm L}_{a,b}
 \le\exp\!\left[
   \frac{s}{216}\bigl(\log2+\log(216e)+\log3\bigr)
 \right]
 =e^{\nu s}.
\]

The bound $e>27/10$ implies $\log216<6$. Hence
$\nu<(7/10+7+11/10)/216=44/1080$.  Combining this with $h>137/1080$ gives
the positive rational margin
\[
 \frac{215}{432}h-\nu-\frac{37}{100000}
 >\frac{137}{1080}\frac{215}{432}
   -\frac{44}{1080}-\frac{37}{100000}
 =\frac{6421483}{291600000}>\frac1{46}.
\]
Since $21/100<\log2$ and $e^{-1/46}<46/47$, every tail charge satisfies
\begin{equation}\label{eq:unequal-tail-charge}
 \widehat x_{a,b}
 <2\cdot952^2\left(\frac{46}{47}\right)^s.
\end{equation}

For each total $s$, fewer than $3s/216$ ordered pairs satisfy
$216|a-b|<s$. Applying the exact first- and second-moment geometric formulas
with $q=46/47$ from $s=4001$ onward shows that the tails contributed to
$\widehat R_0$ and $\widehat R_1$ are, respectively, smaller than
$10^{-27}$ and $10^{-24}$. These are direct rational comparisons using
\[
 \sum_{s\ge N}s^2q^s
 =q^N\left(
   \frac{N^2}{1-q}+\frac{2Nq}{(1-q)^2}
   +\frac{q(1+q)}{(1-q)^3}
 \right).
\]
In particular all tail charges are below $1/6000$. Finally,
\cref{eq:unequal-finite-zero,eq:unequal-finite-one} and these tail bounds
give
\[
\widehat R_0<\frac{364756}{10^9}<\frac{37}{100000}
\]
and, using the finer finite bounds before rounding their sum,
\[
 109\widehat R_0+\widehat R_1
 <\frac{203468108}{10^9}<\frac{21}{100},
\]
as required.
\end{proof}

\begin{lemma}[Elementary binomial and monotonicity estimates]
\label{lem:binomial-estimates}
Let $\ell$ be a positive integer.
\begin{enumerate}[label=\textup{(\roman*)}]
  \item For every integer $1\le d\le\ell$,
  $\binom{\ell}{d}\le(e\ell/d)^d$.
  \item The function $f(u)=u\log(e\ell/u)$ is strictly increasing on
  $0<u<\ell$.
  \item For every real $z\ge1$, one has $z+1\le2^z$.
\end{enumerate}
\end{lemma}

\begin{proof}
For (i), use $\binom{\ell}{d}\le\ell^d/d!$ and the standard bound
$d!\ge(d/e)^d$.  Part (ii) follows from
$f'(u)=\log(\ell/u)>0$.  For (iii), the function
$z\log2-\log(z+1)$ vanishes at $z=1$ and has positive derivative for
$z\ge1$, using \cref{lem:elementary-logs}.
\end{proof}

\begin{lemma}[Deletion-ball exponent]\label{lem:volume-arithmetic}
For $\eta=1/1000$,
\[
  \eta\bigl(\log2+2\log(e/\eta)+\log3\bigr)<\frac9{500}.
\]
\end{lemma}

\begin{proof}
Because $1/\eta=1000=10^3$, we have
$\log(e/\eta)=1+3\log10<8$.  Together with
\cref{lem:elementary-logs}, this gives
\[
  \eta\bigl(\log2+2\log(e/\eta)+\log3\bigr)
  <\frac{7/10+16+11/10}{1000}
  =\frac{89}{5000}<\frac9{500}.
\]
\end{proof}

\begin{lemma}[Event-decay exponent]\label{lem:event-arithmetic}
For every integer $\ell>1000$,
\[
  4\exp\!\left[-\left(\frac{\log2}{18}-\frac9{500}\right)\ell\right]
  <e^{-9\ell/500}.
\]
\end{lemma}

\begin{proof}
Set $\alpha=(\log2)/18-9/500$.  After taking logarithms, it suffices to
show $\alpha-(\log4)/\ell>9/500$.  Since $\log2>69/100$, $\log4<2$, and
$\ell>1000$, we have $(\log4)/\ell<1/500$.  Consequently,
\[
  \alpha-\frac{\log4}{\ell}
  >\frac{23}{600}-\frac9{500}-\frac1{500}
  =\frac{11}{600}.
\]
The last rational is larger than $9/500$, because
$11\cdot500=5500>5400=9\cdot600$.  This proves the logarithmic inequality and hence
the displayed exponential inequality.
\end{proof}

\begin{lemma}[Geometric series and first moment]
\label{lem:geometric-first-moment}
For $0<q<1$ and an integer $L\ge0$,
\begin{align*}
  \sum_{r\ge L}q^r
    &=\frac{q^L}{1-q},\\
  \sum_{r\ge L}r q^r
    &=q^L\left(\frac{L}{1-q}+\frac{q}{(1-q)^2}\right).
\end{align*}
\end{lemma}

\begin{proof}
The first identity is the geometric-series formula.  For the second, write
$r=L+j$ and use
$\sum_{j\ge0}q^j=(1-q)^{-1}$ and
$\sum_{j\ge0}jq^j=q(1-q)^{-2}$.
\end{proof}

\begin{lemma}[Geometric neighbor tail]\label{lem:neighbor-arithmetic}
Let $q_0=e^{-1/80}$.  For every $\ell\ge1000$,
\[
  3\sum_{r\ge1000}(\ell+r)q_0^r<\frac{\ell}{490}.
\]
\end{lemma}

\begin{proof}
The positive exponential series gives
$e^{1/80}>1+1/80=81/80$.  Taking reciprocals yields
$q_0<80/81$.  Hence $1-q_0>1/81$ and
$(1-q_0)^{-1}<81$.  By \cref{lem:geometric-first-moment}, with
$L=1000$ and $q_0^L=e^{-25/2}$,
\begin{align*}
  \sum_{r\ge L}q_0^r
    &=\frac{q_0^L}{1-q_0}<81e^{-25/2},\\
  \sum_{r\ge L}r q_0^r
    &=q_0^L\left(
      \frac{L}{1-q_0}+
      \frac{q_0}{(1-q_0)^2}
    \right)\\
    &<(81L+6561)e^{-25/2}.
\end{align*}
In the second estimate we used $q_0<1$ and
$(1-q_0)^{-2}<81^2=6561$.  Since $\ell\ge L=1000$,
$81L\le81\ell$ and $6561<7\ell$.  Therefore
\[
  3\sum_{r\ge L}(\ell+r)q_0^r
  <3(81+81+7)e^{-25/2}\ell
  =507e^{-25/2}\ell.
\]
The terms through degree six of the positive exponential series give
\[
  e^{5/2}>
  \sum_{j=0}^{6}\frac{(5/2)^j}{j!}
  =\frac{110681}{9216}
  >12,
\]
because $12\cdot9216=110592<110681$.  Consequently
$e^{25/2}=(e^{5/2})^5>12^5=248832$.  Finally,
$507\cdot490=248430<248832<e^{25/2}$, which gives
$507e^{-25/2}<1/490$ and proves the claim.
\end{proof}

\begin{lemma}[Products of small complementary charges]
\label{lem:small-charge-product}
If $0\le u\le1/2$, then $\log(1-u)\ge-2u$.  Consequently, for any finite
family $(u_i)_{i\in I}$ with
$0\le u_i\le1/2$,
\[
  \prod_{i\in I}(1-u_i)
  \ge\exp\!\left(-2\sum_{i\in I}u_i\right).
\]
\end{lemma}

\begin{proof}
Let $\phi(u)=\log(1-u)+2u$.  On $[0,1/2]$,
$\phi'(u)=(1-2u)/(1-u)\ge0$ and $\phi(0)=0$, so
$\log(1-u)\ge-2u$.  Summing this inequality over $i\in I$ and
exponentiating gives the product bound.
\end{proof}

\begin{lemma}[Local-lemma and synchronization margins]
\label{lem:final-margins}
For $\eta=1/1000$, we have $9/500>1/80+1/245$ and
$\eta/(2+\eta)=1/2001>1/2002$.
\end{lemma}

\begin{proof}
The first difference is $9/500-1/80-1/245=139/98000>0$.
Substituting $\eta=1/1000$ gives $\eta/(2+\eta)=1/2001$, and the second
difference is $1/(2001\cdot2002)>0$.
\end{proof}

%% file: appendix/G_long_distance.tex
\section{Distant intervals and bounded event lengths}
\label{app:long-distance}

The adjacent-event local lemma has enough slack to include constraints on
distant intervals.  This also permits a logarithmic cutoff on event lengths.
The matching-path reduction behind the cutoff is the same one used by
Haeupler and Shahrasbi~\cite[Lemma~4.3]{HaeuplerShahrasbi2018}; we give its
form needed here, including the possibility that the reduced intervals are
no longer adjacent.

Fix an output length $M=54N$, put $k=\lceil\log_2 M\rceil$, and set
$R=10000+1200k$ and $B=2R+4$.  Retain the adjacent bad events from
\cref{prop:direct-unequal-gap} with total length at most $B$.
For every ordered pair of disjoint, nonadjacent intervals $I,J$ whose total
length $s$ lies in $[R,B]$, add the event
$\ED(W[I],W[J])<s/216$.  A distant event depends only on macroblocks meeting
$I\cup J$, not on the blocks in the intervening gap.

\begin{lemma}[Probability of a distant event]\label{lem:distant-probability}
For $s\ge2800$, each distant event has probability less than $e^{-s/50}$.
\end{lemma}

\begin{proof}
Write $a=|I|$, $b=|J|$.  If $|a-b|\ge s/216$, the event is empty.
Otherwise both lengths exceed $215s/432$.  Expose all choices except the
blocks wholly contained in $I$.  This fixes $W[J]$, and the conditional
point mass of $W[I]$ is at most $952^2e^{-ha}$, where
$h=(\log952)/54$.

For a deletion pair with $d=u+v<s/216$, the ambient filling bound is at most
$\binom au\binom bv3^u\le\binom sd3^d$.  Here $d$ and $a-b$ determine
$u,v$, so summing over $d$ is sufficient.  Unlike adjacent intervals,
distant intervals may be exactly equal; the sum includes $d=0$.
The elementary binomial-tail bound
$\sum_{d\le ps}\binom sd\le e^{s\mathsf H(p)}$ for $p\le1/2$ gives
a total at most $\exp(s[\mathsf H(1/216)+\log3/216])$.
For completeness, the binomial-tail bound follows by expanding
$(p+(1-p))^s$: for $d\le ps$, the weight
$p^d(1-p)^{s-d}$ is at least $p^{ps}(1-p)^{(1-p)s}$.

Using $\mathsf H(p)\le p\log(e/p)$, $\log216<6$,
$\log3<11/10$, and $h>137/1080$, the probability is less than
$952^2e^{-s/40}$, since
\[
 \frac{215}{432}\frac{137}{1080}-\frac{81}{2160}>\frac1{40}.
\]
Finally $2\log952<14\le s/200$ for $s\ge2800$, which yields the claim
after averaging over the conditioning.
\end{proof}

\begin{lemma}[The combined local lemma]\label{lem:combined-lll}
The bounded adjacent and distant events satisfy the variable local lemma.
Moser--Tardos resampling of these events has expected $O(M)$ resamplings.
\end{lemma}

\begin{proof}
Keep the adjacent charges $\widehat x_{a,b}$ of
\cref{eq:unequal-charges}, and give a distant event of total length $s$ the
charge $z^s$, where $z=100/101$.  At a fixed $s$ there are at most $M^2s$
distant events: choose the two starts and the first length.  Consequently
their total charge is at most
$Q=M^2\sum_{s\ge R}s z^s<10^{-9}$.
Here is a bound uniform in $M$.  Put $y=201/202$; then $z\le y^2$, so
\[
 Q\le M^2 y^R\sum_{s\ge1}s y^s
 <202^2\exp(2\log M-R/202)<10^{-9}.
\]
Indeed, $k\ge\log M$, $1200/202>2$, $10000/202>35$, and
$202^2e^{-35}<10^{-9}$ (for example, $e^{35}>10^{15}$).
Every distant charge is at most $Q$, so their contribution to any
negative logarithm of a neighboring product is at most $2Q$.

For adjacent neighbors, \cref{lem:unequal-charge-envelopes} gives the
convenient weaker bounds $\widehat R_0<366/10^6$ and
$109\widehat R_0+\widehat R_1<204/1000$.
An adjacent event of total length $s$ therefore has neighboring loss at most
$366s/10^6+204/1000+2Q$.  This is strictly less than
$37s/100000+21/100$, the exponent in its charge.  Its local-lemma inequality
follows from \cref{lem:unequal-event-bound}.

For a distant event of total length $s$, an adjacent event of a fixed
ordered length type of total $r$ can meet either of its two components.
The support-overlap bound gives at most $s+2r+218$ possible starts, including
possible overlap within boundary macroblocks.  The adjacent neighboring
loss is therefore at most
$s\widehat R_0+2\widehat R_1+218\widehat R_0
<366s/10^6+408/1000$.
Since $-\log z=\log(101/100)<1/100$ and $s\ge R\ge10000$, its charge
times its neighboring product is greater than
$\exp(-s/100-366s/10^6-408/1000-2Q)>e^{-s/50}$.
\Cref{lem:distant-probability} supplies the required probability bound.

Finally, there are fewer than $M$ adjacent starts of every ordered length
type.  Moser--Tardos bounds the expected number of resamplings by
$\sum_E x_E/(1-x_E)\le M\widehat R_0+2Q=O(M)$.
The theorem permits any rule selecting a currently violated event, which
will be useful for the cached implementation.
\end{proof}

\begin{lemma}[Truncating a bad alignment]\label{lem:alignment-truncation}
If the bounded event system has no violation, then every adjacent interval
pair has normalized edit distance at least $1/216$, and so does every
ordered disjoint interval pair of total length at least $R$.
\end{lemma}

\begin{proof}
Suppose a pair of total length $s>B$ violates the claimed inequality.
Choose an optimal insertion/deletion alignment, represented by a path whose
steps consume one symbol for a deletion or two for a match.  Cut at the
first vertex whose total consumption is at least $\lfloor s/2\rfloor$.
The two subpaths have total consumptions between
$\lfloor s/2\rfloor-1$ and $\lceil s/2\rceil+1$.
Their costs add to less than $s/216$, so at least one has cost less than
one 216th of its own total consumption.  The edit distance of its two
substrings is no greater than this subpath cost.

Both substrings of that subpath are nonempty: otherwise its cost would
equal its total consumption.  They remain ordered and disjoint, although
they need not remain adjacent.  Since $s>2R+4$, each subpath has total
consumption greater than $R$ and strictly less than $s$.  Repeat this
operation on a violating subpair until its total lies in $[R,B]$.
It is then either an included distant event or an adjacent violation.
For adjacent pairs of total at most $B$, the omitted types are impossible
by the short-scale, imbalance, and zero-probability cases in the proof of
\cref{prop:direct-unequal-gap}.  Thus this last pair is a retained event,
a contradiction.  Pairs already in $[R,B]$ require no splitting, and the
same omitted-type argument covers adjacent pairs of smaller total.
\end{proof}

For a requested length $n$, use $M=54\lceil n/54\rceil$ and take a prefix.
When $n\ge2$, the inequalities $M\le54n$, $\log54<4$, and $\log2>0.69$
give $R<18160+1740\log n<30000\log n$.  All interval inequalities inside
the prefix are inherited.  The case $n=1$ has no pair of nonempty disjoint
intervals.  This proves the long-distance assertion of
\cref{thm:long-distance-construction}.  The threshold depends on the
requested finite length; the infinite-word conclusion in
\cref{cor:infinite-ternary} concerns the adjacent inequalities only.

%% file: main.bbl
\newcommand{\etalchar}[1]{$^{#1}$}
\begin{thebibliography}{BDG{\etalchar{+}}26}

\bibitem[BDG{\etalchar{+}}26]{BhattacharyaEtAl2026}
Sudatta Bhattacharya, Sanjana Dey, Elazar Goldenberg, Mursalin Habib, Bernhard
  Haeupler, {Karthik C. S.}, and Michal Kouck{\'y}.
\newblock Constant rate isometric embeddings of {Hamming} metric into {Edit}
  metric.
\newblock In {\em 53rd International Colloquium on Automata, Languages, and
  Programming (ICALP 2026)}, volume 374 of {\em Leibniz International
  Proceedings in Informatics (LIPIcs)}, pages 32:1--32:19. Schloss
  Dagstuhl--Leibniz-Zentrum f{\"u}r Informatik, 2026.

\bibitem[Bec84]{Beck1984}
J\'{o}zsef Beck.
\newblock An application of {Lov\'{a}sz} local lemma: There exists an infinite
  01-sequence containing no near identical intervals.
\newblock In {\em Finite and Infinite Sets}, volume~37 of {\em Colloquia
  Mathematica Societatis J\'{a}nos Bolyai}, pages 103--107. North-Holland,
  1984.

\bibitem[Bri83]{Brinkhuis1983}
Jan Brinkhuis.
\newblock Non-repetitive sequences on three symbols.
\newblock {\em The Quarterly Journal of Mathematics}, 34(2):145--149, 1983.

\bibitem[CHL{\etalchar{+}}19]{ChengEtAl2019}
Kuan Cheng, Bernhard Haeupler, Xin Li, Amirbehshad Shahrasbi, and Ke~Wu.
\newblock Synchronization strings: Highly efficient deterministic constructions
  over small alphabets.
\newblock In {\em Proceedings of the Thirtieth Annual ACM-SIAM Symposium on
  Discrete Algorithms (SODA)}, pages 2185--2204. SIAM, 2019.

\bibitem[CMR19]{CurrieMolRampersad2019}
James~D. Currie, Lucas Mol, and Narad Rampersad.
\newblock Circular repetition thresholds on some small alphabets: Last cases of
  {Gorbunova's} conjecture.
\newblock {\em The Electronic Journal of Combinatorics}, 26(2):P2.31, 2019.

\bibitem[CR16]{CamungolRampersad2016}
Serina Camungol and Narad Rampersad.
\newblock Avoiding approximate repetitions with respect to the longest common
  subsequence distance.
\newblock {\em Involve, a Journal of Mathematics}, 9(4):657--666, 2016.

\bibitem[EL75]{ErdosLovasz1975}
Paul Erd\H{o}s and L\'{a}szl\'{o} Lov\'{a}sz.
\newblock Problems and results on 3-chromatic hypergraphs and some related
  questions.
\newblock In {\em Infinite and Finite Sets}, volume~10 of {\em Colloquia
  Mathematica Societatis J\'{a}nos Bolyai}, pages 609--627. North-Holland,
  1975.

\bibitem[EZ98]{EkhadZeilberger1998}
Shalosh~B. Ekhad and Doron Zeilberger.
\newblock There are more than $2^{n/17}$ $n$-letter ternary square-free words.
\newblock {\em Journal of Integer Sequences}, 1(1):Article 98.1.9, 1998.

\bibitem[EZ01]{EkhadZeilbergerErratum2001}
Shalosh~B. Ekhad and Doron Zeilberger.
\newblock Errata and addenda to ``there are more than $2^{n/17}$ $n$-letter
  ternary square-free words''.
\newblock Author erratum, 12 June 2001, 2001.

\bibitem[GKN20]{GrytczukKordulewskiNiewiadomski2020}
Jaros{\l}aw Grytczuk, Hubert Kordulewski, and Artur Niewiadomski.
\newblock Extremal square-free words.
\newblock {\em The Electronic Journal of Combinatorics}, 27(1):P1.48, 2020.

\bibitem[Gri01]{Grimm2001}
Uwe Grimm.
\newblock Improved bounds on the number of ternary square-free words.
\newblock {\em Journal of Integer Sequences}, 4:Article 01.2.7, 2001.

\bibitem[Hae19]{Haeupler2019DocumentExchange}
Bernhard Haeupler.
\newblock Optimal document exchange and new codes for insertions and deletions.
\newblock In {\em 60th IEEE Annual Symposium on Foundations of Computer Science
  (FOCS 2019)}, pages 334--347. IEEE, 2019.

\bibitem[HRS19]{HaeuplerRubinsteinShahrasbi2019}
Bernhard Haeupler, Aviad Rubinstein, and Amirbehshad Shahrasbi.
\newblock Near-linear time insertion-deletion codes and
  $(1+\varepsilon)$-approximating {Edit} distance via indexing.
\newblock In {\em Proceedings of the 51st Annual ACM SIGACT Symposium on Theory
  of Computing (STOC)}, pages 697--708. ACM, 2019.

\bibitem[HS17]{HaeuplerShahrasbi2017}
Bernhard Haeupler and Amirbehshad Shahrasbi.
\newblock Synchronization strings: Codes for insertions and deletions
  approaching the {Singleton} bound.
\newblock In {\em Proceedings of the 49th Annual ACM SIGACT Symposium on Theory
  of Computing (STOC)}, pages 33--46. ACM, 2017.

\bibitem[HS18]{HaeuplerShahrasbi2018}
Bernhard Haeupler and Amirbehshad Shahrasbi.
\newblock Synchronization strings: Explicit constructions, local decoding, and
  applications.
\newblock In {\em Proceedings of the 50th Annual ACM SIGACT Symposium on Theory
  of Computing (STOC)}, pages 841--854. ACM, 2018.

\bibitem[HS21]{HaeuplerShahrasbiSurvey2021}
Bernhard Haeupler and Amirbehshad Shahrasbi.
\newblock Synchronization strings and codes for insertions and deletions: A
  survey.
\newblock {\em IEEE Transactions on Information Theory}, 67(6):3190--3206,
  2021.

\bibitem[HSV18]{HaeuplerShahrasbiVitercik2018}
Bernhard Haeupler, Amirbehshad Shahrasbi, and Ellen Vitercik.
\newblock Synchronization strings: Channel simulations and interactive coding
  for insertions and deletions.
\newblock In {\em 45th International Colloquium on Automata, Languages, and
  Programming (ICALP 2018)}, volume 107 of {\em Leibniz International
  Proceedings in Informatics (LIPIcs)}, pages 75:1--75:14. Schloss
  Dagstuhl--Leibniz-Zentrum f{\"u}r Informatik, 2018.

\bibitem[MR21]{MolRampersad2021}
Lucas Mol and Narad Rampersad.
\newblock Lengths of extremal square-free ternary words.
\newblock {\em Contributions to Discrete Mathematics}, 16(1):8--19, 2021.

\bibitem[MT10]{MoserTardos2010}
Robin~A. Moser and G\'{a}bor Tardos.
\newblock A constructive proof of the general {Lov\'{a}sz} local lemma.
\newblock {\em Journal of the ACM}, 57(2):11:1--11:15, 2010.

\bibitem[SDL16]{SollamiDouglasLiebmann2016}
Michael Sollami, Craig~C. Douglas, and Manfred Liebmann.
\newblock An improved lower bound on the number of ternary squarefree words.
\newblock {\em Journal of Integer Sequences}, 19:Article 16.6.7, 2016.

\bibitem[Shu10]{Shur2010Circular}
Arseny~M. Shur.
\newblock On ternary square-free circular words.
\newblock {\em The Electronic Journal of Combinatorics}, 17(1):R140, 2010.

\bibitem[Thu06]{Thue1906}
Axel Thue.
\newblock \"{U}ber unendliche zeichenreihen.
\newblock {\em Skrifter udgivne af Videnskabsselskabet i Christiania, I.
  Mathematisk-naturvidenskabelig Klasse, No. 7}, pages 1--22, 1906.

\end{thebibliography}
